\documentclass[3p,review]{elsarticle}

\usepackage{lineno,hyperref,graphicx,epstopdf}
\usepackage{times}
\usepackage{type1cm}
\fontsize{12pt}{18pt}

\usepackage{braket}
\usepackage{amsmath}
\usepackage{amsfonts}
\usepackage{amssymb}
\usepackage{amsthm}
\usepackage{algorithm}
\usepackage{algpseudocode}
\usepackage {braket}
\usepackage{diagbox}
\usepackage{makecell}

\theoremstyle{plain}
\newtheorem{theorem}{Theorem}

\newtheorem{lemma}{Lemma}

\newtheorem{corollary}{Corollary}

\theoremstyle{definition}
\newtheorem{defi}{Definition}

\newtheorem{remark}{Remark}

\newtheorem{example}{Example}

\journal{}

\date{}

\begin{document}
\begin{frontmatter}
\title{Distributed exact quantum algorithms for Deutsch-Jozsa problem}

\author{Hao Li$^{1,2}$}
\author{Daowen Qiu$^{1,2,4}$\corref{one}}
\author{Le Luo$^{3,4}$}
\cortext[one]{Corresponding author (D. Qiu). {\it E-mail addresses:} issqdw@mail.sysu.edu.cn (D. Qiu)}

\address{
   $^{1}$Institute of Quantum Computing and Computer Theory, School of Computer Science and Engineering, Sun Yat-sen University, Guangzhou 510006, China\\
   $^{2}$The Guangdong Key Laboratory of Information Security Technology, Sun Yat-sen University, 510006, China\\
     $^{3}$School of Physics and Astronomy, Sun Yat-sen University, Zhuhai 519082, China\\
 $^{4}$QUDOOR Technologies Inc., Guangzhou,  China }

\begin{abstract}

Deutsch-Jozsa (DJ) problem is one of the most important problems demonstrating the power of quantum algorithm. DJ problem can be described as a Boolean function $f$: $\{0,1\}^n\rightarrow \{0,1\}$ with promising it is either constant or balanced, and the purpose is to determine which type it is. DJ algorithm can solve it exactly with one query.   In this paper, we first discover the inherent structure of  DJ problem in distributed scenario by giving a number of equivalence characterizations between $f$ being constant (balanced) and some properties of $f$'s subfunctions, and then we  propose three distributed exact quantum algorithms for solving DJ problem. 
Our algorithms have essential acceleration over distributed classical deterministic algorithm, and
  can be extended to the case  of multiple computing nodes.  Compared with  DJ algorithm, our algorithms can reduce the number of qubits and the depth of circuit implementing a single query operator. Therefore, we find that the structure of problem should be clarified for designing distributed quantum algorithm to solve it.

\end{abstract}

\begin{keyword}
Deutsch-Jozsa problem \sep Distributed quantum algorithm \sep Circuit depth  
\end{keyword}

\end{frontmatter}


\section{Introduction}\label{sec:introduction}

Quantum computing \cite{nielsen_quantum_2010
} has been proved to have great potential in factorizing  large numbers \cite{shor_polynomial-time_1997},  searching unordered database \cite{grover_fast_1996} and solving linear systems of equations \cite{HHL_2009}. 
Due to the limitations of current physical technology, large-scale universal quantum computers have not been realized. At present, quantum technology has been entered to the Noisy Intermediate-scale Quantum (NISQ) era \cite{preskill_quantum_2018}, it is possible for us to implement quantum algorithms on middle-scale quantum circuits.

Distributed quantum computing is an innovative computing architecture, which combines quantum computing with distributed computing \cite{goos_distributed_2003,beals_efficient_2013,Qiu2017DQC,Tan2022DQCSimon,Qiu22,Xiao2023DQAShor}. In distributed quantum computing, multiple  computing nodes communicate with each other and cooperate to complete computing tasks, the circuit size and depth can be reduced
, which  is beneficial to reduce the noise of circuits. 


As mentioning above, DJ problem can be described as a Boolean function $f:\{0,1\}^n\rightarrow \{0,1\}$ with promising it is either constant or balanced, and the purpose is to decide what function  it is. DJ algorithm can solve it exactly with one query, but it is easy to know that the query complexity of classical deterministic algorithm is $2^{n-1}+1$ in the worst case.
DJ problem is the first  to show that quantum algorithm has essential acceleration compared with classical deterministic algorithm \cite{deutsch_rapid_1992}, which is an extension of  Deutsch's problem \cite{deutsch_quantum_1985}.  Furthermore, Qiu and Zheng \cite{Generalized Deutsch-Jozsa proble_2018} extended DJ problem to  generalized DJ problem, and gave an optimal algorithm for  generalized DJ problem. DJ algorithm  \cite{deutsch_rapid_1992} presents a basic framework of quantum algorithms, and in a way provides inspiration for Simon’s algorithm \cite{simon_power_1997}, Shor’s algorithm \cite{shor_polynomial-time_1997} and Grover’s algorithm \cite{grover_fast_1996}. In particular, Qiu and Zheng  \cite{Revisiting Deutsch-Jozsa algorithm_2018} proved that  DJ algorithm can compute any symmetric partial Boolean function with exact quantum 1-query complexity.
 In fact, DJ algorithm can be used to test some properties of Boolean functions \cite{Testing Boolean Functions Properties_2021}.
In addition, it is worth mentioning that  there is also relevant  physical realization on  DJ algorithm \cite{Implementation of the Deutsch-Jozsa Algorithm_2013}.

In this paper, we first discover the inherent structure of DJ problem in distributed scenario by means of presenting a number of equivalence characterizations concerning the features of DJ problem.  More exactly,
let DJ problem be represented by a Boolean function $f:\{0,1\}^n\rightarrow\{0,1\}$, $f$ is decomposed  into $2^t$ subfunctions $f_w:\{0,1\}^{n-t}\rightarrow\{0,1\}$, where $f_w(u)=f(uw)$ $(u\in\{0,1\}^{n-t}, w\in\{0,1\}^t)$.
Then we prove a number of equivalence relationships between $f$ being constant (balanced) and some properties of its subfunctions.
In particular, we
 design three distributed exact quantum algorithms  for solving it, while we utilize quantum teleportation technology \cite{bennett_teleporting_1993} to realize quantum gates spanning multiple computing nodes \cite{caleffi_quantum_2018}. 
In our  algorithms, after the oracle queries for multiple computing nodes, the specific unitary operator we design is used to process the query results. Eventually,  the quantum state containing the structure consistent with DJ problem is obtained, and then the solution of DJ problem is concluded exactly by measuring this state.

Actually, we design three distributed exact quantum algorithms to solve DJ problem, namely Algorithm \ref{algorithm3}, Algorithm \ref{algorithm2} and Algorithm \ref{algorithm5}. The design of Algorithm \ref{algorithm3} 
is to associate the two subfunctions $f_0$ and $f_1$ of  $f$  with xor operation, extracting the value of $f_0$ by using the Pauli operator $Z$. However, this algorithm can only solve DJ problem in distributed scenario with two subfunctions.

The design of Algorithm \ref{algorithm2} draws on the ideas and methods of distributed Simon's algorithm by Tan, Xiao, and Qiu et al \cite{Tan2022DQCSimon} and HHL algorithm \cite{HHL_2009}. In fact, Algorithm \ref{algorithm2} uses the specific unitary operator to perform addition and subtraction operations on the values of multiple subfunctions $f_w$ $(w\in\{0,1\}^t)$ of $f$,  and then  integrates the specific controlled rotation operator from HHL algorithm \cite{HHL_2009} to extract the values of the joint operation of the multiple subfunctions.

  Algorithm \ref{algorithm5} combines the  ideas and methods of Algorithm \ref{algorithm3} and Algorithm \ref{algorithm2}, which can solve DJ problem in distributed scenario with multiple computing nodes, and some of its single quantum gates require less qubits than Algorithm \ref{algorithm2}. Algorithm \ref{algorithm5} also draws on the ideas of distributed Simon's algorithm \cite{Tan2022DQCSimon} and HHL algorithm \cite{HHL_2009}.

The algorithms we  design have the following advantages. First, compared with the distributed quantum algorithm for solving  DJ problem in \cite{avron_quantum_2021}, our algorithms are exact and have better scalability. Second, compared with distributed classical deterministic  algorithm, our algorithms have  exponential advantage in query complexity. Finally, compared with  DJ algorithm, the single query operator in our algorithms requires fewer qubits, and the depth of circuit is reduced \cite{Elementary_gates_1995}, which has better anti-noise performance.

The remainder of the paper is organized as follows.  
In Section \ref{sec:preliminaries}, after briefly recalling  DJ problem and DJ algorithm, we  describe DJ problem in distributed scenario and  review the distributed  DJ  algorithm with errors  \cite{avron_quantum_2021} that motivates the study of designing distributed exact DJ algorithm. 
In Section \ref{sec:Characterization of the structure of DJ problem in distributed scenario}, we study a number of properties of DJ problem in  distributed scenario in detail, which are key to designing distributed exact DJ algorithms.
In section \ref{sec:Distributed quantum algorithm for  DJ problem (Algorithm 1)}, we first describe the ideas and methods for designing distributed exact DJ algorithms, and then  give  Algorithm \ref{algorithm3}. 
 Algorithm \ref{algorithm3} can solve  DJ problem in  distributed scenario with only  two computing nodes.
In section \ref{sec:Distributed quantum algorithm for  DJ problem (Algorithm 2)},  we describe  Algorithm \ref{algorithm2}. Algorithm \ref{algorithm2} can solve  DJ problem in  distributed scenario with multiple computing nodes, and its  key technique is to use the controlled rotation operator in  HHL algorithm \cite{HHL_2009} to extract the information related to  the structure of  DJ problem.
In section \ref{sec:Distributed quantum algorithm for  DJ problem (Algorithm 3)}, we describe  Algorithm \ref{algorithm5} and also compare the three algorithms we propose. Also, Algorithm \ref{algorithm5}  can solve  DJ problem in  distributed scenario with multiple computing nodes,  and   fewer qubits in this algorithm are required to implement certain unitary operators  compared to Algorithm \ref{algorithm2}. Finally,
in Section \ref{sec:conclusions}, we conclude with a summary. 


\section{Preliminaries}\label{sec:preliminaries}
In this section, we first recall DJ problem and DJ algorithm, and then describe DJ problem in distributed scenario as well as review the distributed  DJ  algorithm with errors \cite{avron_quantum_2021}.  We assume that the readers are familiar with  liner algebra, probability theory and
basic notations in quantum computing \cite{nielsen_quantum_2010}.

\subsection{DJ problem and DJ algorithm}
  DJ problem can be described as follows: Consider a Boolean function $f:\{0,1\}^n \rightarrow \{0,1\}$, where we have the promise that $f$ is either constant or balanced. 

Call $f$ is constant if and only if 
\begin{equation}
f(x) \equiv 0
\end{equation}
 or 
 \begin{equation}
 f(x) \equiv 1. 
\end{equation}

Call $f$ is balanced if and only if  
\begin{equation}
|\{x|f(x)=0\}|=|\{x|f(x)=1\}|
=2^{n-1}. 
\end{equation}

Suppose we have an oracle $O_f$ that can query the value of Boolean function $f$, where the  oracle $O_f$ is defined as:
\begin{align}
O_{f}\ket{x}\ket{b}=&\ket{x}\ket{b\oplus f(x)},
\end{align}
 where $x\in\{0,1\}^{n}$ and $b\in\{0,1\}$.

The goal of   DJ problem is to determine whether $f$ is constant or balanced.



 DJ algorithm is  the first quantum algorithm that is essentially faster than
any possible deterministic classical algorithm for solving DJ problem. 
In DJ algorithm, it first generates a uniform superposition state by acting on an initial state with  Hadamard transform, and then the query operator $O_f$ and   Hadamard transform are applied sequently, finally it performs measurement to determine whether $f$ is constant or balanced. The details of DJ algorithm are described in \ref{DJ Algorithm}.

\subsection{DJ problem in distributed scenario}
In the following, we first describe DJ problem in distributed scenario with two distributed computing nodes,  then describe DJ problem in distributed scenario with multiple distributed computing nodes.

In the case of two distributed computing nodes,  Boolean function $f$ corresponding to DJ problem is divided into two subfunctions $f_0:\{0,1\}^{n-1}\rightarrow\{0,1\}$ and $f_1:\{0,1\}^{n-1}\rightarrow\{0,1\}$ as follows.
For all $ u \in \{0,1\}^{n-1}$, let
\begin{equation}
f_0(u)=f(u0)
\end{equation}
and
\begin{equation}
f_1(u)=f(u1).
\end{equation}

  Suppose Alice has an oracle $O_{f_0}$ that can query $f_0(u)$ for all $u \in \{0,1\}^{n-1}$,  Bob has an oracle $O_{f_1}$ that can query  $f_1(u)$ for all $u \in \{0,1\}^{n-1}$, 
where the oracle $O_{f_0}$ and $O_{f_1}$  are defined as:
\begin{align}
O_{f_0}\ket{u}\ket{b}=&\ket{u}\ket{b\oplus f_0(u)},\label{Of0}\\
O_{f_1}\ket{u}\ket{b}=&\ket{u}\ket{b\oplus f_1(u)},\label{Of1}
\end{align}
 where $u\in\{0,1\}^{n-1}$ and $b\in\{0,1\}$.

They need to determine whether $f$ is constant or balanced by querying their own oracle and communicating with each other as few times as possible. 


In the case of multiple distributed computing nodes,  Boolean function $f$ corresponding to DJ problem is divided into $2^t$ subfunctions $f_w:\{0,1\}^{n-t}\rightarrow\{0,1\}$ as follows.
For all $ u \in \{0,1\}^{n-t}$, let
\begin{equation}\label{General method of function decomposition}
f_w(u)=f(uw),
\end{equation}
where $w\in\{0,1\}^t$. 


Suppose there are $2^t$ people, each of whom has an oracle $O_{f_w}$ that can query all $f_w(u)=f(uw)$ for all $u \in \{0,1\}^{n-t}$, $w \in \{0,1\}^t$, 
where the oracle $O_{f_w}$ is defined as:
\begin{equation}\label{Ofw}
O_{f_w}\ket{u}\ket{b}=\ket{u}\ket{b\oplus f_w(u)},
\end{equation}
where $ w\in \{0,1\}^t$, $u\in\{0,1\}^{n-t}$, $b\in\{0,1\}$.

Each person can access $2^{n-t}$ values of $f$. 
They need to determine whether $f$ is constant or balanced by querying their own oracle and communicating with each other as few times as possible.


\subsection{Distributed  DJ  algorithm with errors}

In 2021, Avron et al \cite{avron_quantum_2021} proposed a distributed DJ  algorithm. However, their algorithm only considers the case of  two computing nodes, and  it directly runs DJ algorithm for each node (i.e., subfunction) without using quantum communication between nodes. In this way, the result of measurement in their algorithm is not exact and has error. 

In the following we review the algorithm in \cite{avron_quantum_2021}  and give its error analysis.

Given a Boolean function $f:\{0,1\}^n\rightarrow\{0,1\}$ of  DJ problem,  and it is decomposed  into two subfunctions $f_0:\{0,1\}^{n-1}\rightarrow\{0,1\}$ and $f_1:\{0,1\}^{n-1}\rightarrow\{0,1\}$. The algorithm in \cite{avron_quantum_2021} runs the DJ algorithm directly for $f_0$ and $f_1$ respectively, obtaining the measuring  results denoted as $M_0$ and $M_1$ (if $M_0=M_1=0^{n-1}$, then $f$ is concluded to be constant), respectively. Let $k_0=|\{u\in\{0,1\}^{n-1}|f_0(u)=1\}|$, $k_1=|\{u\in\{0,1\}^{n-1}|f_1(u)=1\}|$, $N=2^n$.

The  probability that the algorithm in \cite{avron_quantum_2021} misidentifies a balanced function as constant is
\begin{equation}\label{DDJA_error_2node_probability}
\begin{split}
&\Pr(f\ is\ \ balanced,M_0=M_1=0^{n-1})\\
=&\sum_{\substack{k_0+k_1=\frac{N}{2}\\ 0\leq k_0\leq\frac{N}{2}\\0\leq k_1\leq\frac{N}{2}  }}\dfrac{\dbinom{N/2}{k_0}\dbinom{N/2}{k_1}}{\dbinom{N}{N/2}}\left(\dfrac{2^{2}}{N}\right)^2\left(k_0-\dfrac{N}{2^{2}}\right)^2\left(\dfrac{2^{2}}{N}\right)^2\left(k_1-\dfrac{N}{2^{2}}\right)^2\\
>&0.
\end{split}
\end{equation}

From equation (\ref{DDJA_error_2node_probability}), it is clear that there is a situation where the algorithm in \cite{avron_quantum_2021} has errors. In effect, the algorithm in \cite{avron_quantum_2021} does not consider exploiting the essential structure of  DJ problem.
In \ref{Distributed DJ algorithm for multiple computing nodes with errors}, we generalize the algorithm in \cite{avron_quantum_2021} to solve the case of multiple subfunctions and analyze its errors. 

\section{Characterizations of DJ problem in distributed scenario} \label{sec:Characterization of the structure of DJ problem in distributed scenario}


In this section, we characterize the essential structure and features of  DJ problem in  distributed scenario.
 Intuitively, the structure of  DJ problem in distributed scenario can be represented from its corresponding structure table, and we construct   related examples in \ref{Examples of the structure of  DJ problem in  distributed scenario}.

In the interest of simplicity, we give  a number of  notations.
Suppose Boolean function $f:\{0,1\}^n \rightarrow \{0,1\}$, $f_0:\{0,1\}^{n-1} \rightarrow \{0,1\}$, $f_1:\{0,1\}^{n-1} \rightarrow \{0,1\}$,  where $f_0(u)=f(u0)$, $f_1(u)=f(u1)$,  $u\in\{0,1\}^{n-1}$, denote by 

\begin{align}
C_{00}=&|\{u|f_0(u)=0\}|.\\
C_{01}=&|\{u|f_0(u)=1\}|.\\
C_{10}=&|\{u|f_1(u)=0\}|.\\
C_{11}=&|\{u|f_1(u)=1\}|.\\
B_{00}=&|\{u|f_0(u)=0, f_1(u)=0\}|.\\
B_{01}=&|\{u|f_0(u)=0, f_1(u)=1\}|.\\
B_{10}=&|\{u|f_0(u)=1, f_1(u)=0\}|.\\
B_{11}=&|\{u|f_0(u)=1, f_1(u)=1\}|.\\
 M=&|\{u|f_0(u)\oplus f_1(u)=0\}|.
\end{align}

It is clear that 
\begin{equation}\label{M=B_0+B_1}
M=B_{00}+B_{11}
\end{equation}
and $0\leq M\leq 2^{n-1}$.


Theorem \ref{The4} below provides a sufficient and necessary condition for determining whether a given Boolean function $f$ of  DJ problem is constant or balanced.

\begin{theorem}\label{The4} Suppose Boolean  function $f:\{0,1\}^n \rightarrow \{0,1\}$, satisfies that it is either constant or balanced, and it is divided into subfunctions $f_0$ and $f_1$. Then:
\begin{enumerate}[(1)]
\item $f$ is constant  if and only if $C_{00}=C_{10}=2^{n-1}$
 or $C_{01}=C_{11}=2^{n-1}$;
\item $f$ is balanced if and only if $B_{00}=B_{11}=M/2$. 
 \end{enumerate}
\end{theorem}
\begin{proof}

Firstly, we prove that $f$ is constant  if and only if $C_{00}=C_{10}=2^{n-1}$
 or $C_{01}=C_{11}=2^{n-1}$. 

(\romannumeral1) $\Longleftarrow$. If $C_{00}=C_{10}=2^{n-1}$
 or $C_{01}=C_{11}=2^{n-1}$,
then 
\begin{equation}
|\{x|f(x)=1\}|=0
\end{equation}
or 
\begin{equation}
|\{x|f(x)=1\}|=2^n,
\end{equation}
that is $f(x) \equiv 0$ or $f(x) \equiv 1$. Therefore, $f$ is constant.

(\romannumeral2) $\Longrightarrow$. If $f$ is constant, then  $f(x) \equiv 0$ or $f(x) \equiv 1$. 

So
\begin{equation}
 |\{x|f(x)=1\}|=0
\end{equation}
or
\begin{equation}
 |\{x|f(x)=1\}|=2^n. 
\end{equation}

Therefore 
\begin{equation}
C_{00}=C_{10}=2^{n-1}
\end{equation}
or 
\begin{equation}
C_{01}=C_{11}=2^{n-1}.
\end{equation}

Secondly, we prove that $f$ is balanced if and only if $B_0=B_1=M/2$.

(\romannumeral3) 
$\Longleftarrow$. 
Suppose $B_{00}=B_{11}=M/2$ $(0\leq M\leq 2^{n-1})$, then $f$ is constant. 

If $f$ is constant, then $f\equiv 0$ or $f\equiv 1$. 
So $B_{00}=2^{n-1}$ or $B_{11}=2^{n-1}$, which is contrary to the assumption $B_{00}=B_{11}=M/2$ $(0\leq M\leq 2^{n-1})$. 

Therefore, if $B_{00}=B_{11}=M/2$, then $f$ is balanced.

(\romannumeral4) $\Longrightarrow$. 
Since 
\begin{equation}\label{thm4equ6}
\begin{split}
&B_{00}
+B_{11}
+B_{01}+B_{10}\\
=&M+B_{01}+B_{10}\\
=&2^{n-1},
 \end{split}
\end{equation}
 we have 
\begin{equation}\label{thm4equ3}
B_{01}+B_{10}=2^{n-1}-M. 
\end{equation}


Denote
 \begin{equation}\label{S_0_eq}
S_0=C_{00}-B_{01},
\end{equation}
\begin{equation}\label{S_1_eq}
S_1=C_{10}-B_{10}.
\end{equation}

Then 
\begin{equation}\label{thm4equ4}
S_0=S_1
=B_{00}.
\end{equation}

According to equation (\ref{thm4equ3}),equation (\ref{S_0_eq}) and equation (\ref{S_1_eq}), we have
\begin{equation}\label{thm2_eq_S0+S1}
\begin{split}
S_0+S_1
=&C_{00}+C_{10}-(B_{01}+B_{10})\\
=&C_{00}+C_{10}-\left(2^{n-1}-M\right).
\end{split}
\end{equation}

If $f$ is balanced, 
then
\begin{equation}\label{thm2_eq_S0+S1_f(x)=0}
C_{00}+C_{10}=2^{n-1}. 
\end{equation}

Therefore, by equation (\ref{thm2_eq_S0+S1}) and equation (\ref{thm2_eq_S0+S1_f(x)=0}), we have
\begin{equation}\label{thm4equ5}
\begin{split}
S_0+S_1
=&2^{n-1}-\left(2^{n-1}-M\right)\\
=&M.
\end{split}
\end{equation}

According to equation (\ref{thm4equ4}) and equation (\ref{thm4equ5}), we have
\begin{equation}\label{thm4equ1}
S_0=S_1
=B_{00}
=M/2. 
\end{equation}


With equation  (\ref{M=B_0+B_1}) and equation (\ref{thm4equ1}), we have
\begin{equation}\label{thm4equ2}
\begin{split}
B_{11}=&M-B_{00}\\
=&M-M/2\\
=&M/2.
\end{split}
\end{equation}

Combining equation (\ref{thm4equ1}) with equation (\ref{thm4equ2}), we have $B_{00}=B_{11}=M/2$.






\end{proof}

In light of Theorem \ref{The4},  we have   Corollary \ref{Cor3} below, which provides a  sufficient  condition for determining whether a given Boolean function $f$ of  DJ problem is balanced by  the xor value of the subfunctions $f_0$ and $f_1$.


\begin{corollary}\label{Cor3} Suppose Boolean function $f:\{0,1\}^n \rightarrow \{0,1\}$, satisfies that it is either constant or balanced, and it is divided into subfunctions $f_0$ and $f_1$. If $\exists$ $u\in\{0,1\}^{n-1}$ such that $f_0(u)\oplus f_1(u)=1$, then $f$ is balanced.
\end{corollary}


In the following, we  describe 
Theorem \ref{The2}, which characterizes the basic structure of  DJ problem in  distributed scenario with multiple subfunctions. Also, we give some  notations in order to describe its procedure of proof more clearly.


Suppose Boolean function $f:\{0,1\}^n \rightarrow \{0,1\}$, $f_w:\{0,1\}^{n-t} \rightarrow \{0,1\}$, where $f_w(u)=f(uw)$, $u\in\{0,1\}^{n-t}$, $w\in\{0,1\}^t$. For all $u \in \{0,1\}^{n-t}$, denote
\begin{equation}\label{deltadef}
 \delta(u) = |\{w|f_w(u)=0\}|-|\{w|f_w(u)=1\}|.
 \end{equation}
 

According to  equation (\ref{deltadef}) 
, we have
\begin{equation}\label{delta1}
\begin{split}
\delta(u)=&|\{w|f_w(u)=0\}|-|\{w|f_w(u)=1\}|\\
  =&2^t-|\{w|f_w(u)=1\}|-|\{w|f_w(u)=1\}|\\
  =&2^t-2|\{w|f_w(u)=1\}|\\ 
  =&2^t-2\sum_{w\in\{0,1\}^t}f_w(u).
\end{split}
\end{equation}

  Theorem \ref{The2} below provides a sufficient and necessary condition for determining whether a given Boolean function $f$ of  DJ problem is constant or balanced by means of $\delta(u)$.


\begin{theorem}\label{The2} Suppose Boolean function $f:\{0,1\}^n \rightarrow \{0,1\}$, satisfies that it is either constant or balanced, and it is divided into $2^t$ subfunctions $f_w$ $(\forall u \in \{0,1\}^{n-t}, w \in \{0,1\}^{t}, f_w(u)=f(uw))$. Then: 
\begin{enumerate}[(1)]
\item $f$ is constant  if and only if for all $u \in \{0,1\}^{n-t}$, $\delta(u) =2^t$ or for all $u \in \{0,1\}^{n-t}$, $\delta(u) =-2^t$; 
\item $f$ is balanced if and only if $\sum\nolimits_{u\in\{0,1\}^{n-t}} \delta(u) = 0$. 
\end{enumerate}
\end{theorem}
\begin{proof}
Firstly, we prove that $f$ is constant if and only if for all $u \in \{0,1\}^{n-t}$, $\delta(u) =2^t$ or for all $u \in \{0,1\}^{n-t}$, $\delta(u) =-2^t$.

(\romannumeral1) $\Longleftarrow$. If  for all $u \in \{0,1\}^{n-t}$, $\delta(u) =2^t$, then according to equation (\ref{delta1}), for all $u \in \{0,1\}^{n-t}$, we have 
\begin{equation}
\sum_{w\in\{0,1\}^t}f_w(u)=0. 
\end{equation}

So for all $x \in \{0,1\}^{n}$, $f(x)=0$, that is $f(x) \equiv 0$.

 If  for all $u \in \{0,1\}^{n-t}$, $\delta(u) =-2^t$, then according to equation (\ref{delta1}), for all $u \in \{0,1\}^{n-t}$, we have 
 \begin{equation}
 \sum_{w\in\{0,1\}^t}f_w(u)=2^t. 
 \end{equation}
 
 So for all $x \in \{0,1\}^{n}$, $f(x)=1$, that is $f(x) \equiv 1$.
 
 Therefore, if for all $u \in \{0,1\}^{n-t}$, $\delta(u) =2^t$ or for all $u \in \{0,1\}^{n-t}$, $\delta(u) =-2^t$, then $f$ is constant.

(\romannumeral2) $\Longrightarrow$. If $f$ is constant, then we have $f(x) \equiv 0$ or $f(x) \equiv 1$.

If $f(x) \equiv 0$, then for all $x \in \{0,1\}^{n}$, $f(x)=0$. So for all $u \in \{0,1\}^{n-t}$, we have 
\begin{equation}
\sum_{w\in\{0,1\}^t}f_w(u)=0. 
 \end{equation}
 
 According to equation (\ref{delta1}), for all $u \in \{0,1\}^{n-t}$, we have $\delta(u) =2^t$.

If $f(x) \equiv 1$, then for all $x \in \{0,1\}^{n}$, $f(x)=1$. So for all $u \in \{0,1\}^{n-t}$, we have 
\begin{equation}
\sum_{w\in\{0,1\}^t}f_w(u)=2^t.  
 \end{equation}
 
According to equation (\ref{delta1}), for all $u \in \{0,1\}^{n-t}$, we have $\delta(u) =-2^t$.

 Therefore, if $f$ is constant, then for all $u \in \{0,1\}^{n-t}$, $\delta(u) =2^t$ or for all $u \in \{0,1\}^{n-t}$, $\delta(u) =-2^t$.

Secondly, we prove that $f$ is balanced if and only if $\sum\nolimits_{u\in\{0,1\}^{n-t}} \delta(u) = 0$.

(\romannumeral3) $\Longleftarrow$. If $\sum\nolimits_{u\in\{0,1\}^{n-t}} \delta(u)=0$, according to equation $(\ref{delta1})$, then we have 
\begin{equation}
\begin{split}
\sum\limits_{u\in\{0,1\}^{n-t}} \delta(u)
=&\sum\limits_{u\in\{0,1\}^{n-t}} \left(2^t-2\sum\limits_{w\in\{0,1\}^t}f_w(u)\right)\\
=&\sum\limits_{u\in\{0,1\}^{n-t}}2^t-2\sum\limits_{\substack{u\in\{0,1\}^{n-t}\\w\in\{0,1\}^{t}}}f_w(u)\\
=&2^n-2\sum\limits_{x\in\{0,1\}^{n}}f(x)
=0, 
\end{split}
\end{equation}
that is $\sum\nolimits_{x\in\{0,1\}^{n}}f(x)=2^{n-1}$.
So we have 
\begin{equation}
|\{x|f(x)=1\}|
=\sum\limits_{x\in\{0,1\}^{n}}f(x)
=2^{n-1}. 
\end{equation}

\begin{equation}
|\{x|f(x)=0\}|
=2^n-|\{x|f(x)=1\}|
=2^{n-1}.
\end{equation}

Therefore, $f$ is balanced.

(\romannumeral4) $\Longrightarrow$. If $f$ is balanced, then  we have 
\begin{equation}\label{thm5_balanced_1}
|\{x|f(x)=0\}|=|\{x|f(x)=1\}|=2^{n-1}. 
\end{equation}

According to equation $(\ref{delta1})$ and equation $(\ref{thm5_balanced_1})$,  we have 
\begin{equation}
\begin{split}
\sum\limits_{u\in\{0,1\}^{n-t}} \delta(u)
=&\sum\limits_{u\in\{0,1\}^{n-t}} \left(2^t-2\sum\limits_{w\in\{0,1\}^t}f_w(u)\right)\\
=&\sum\limits_{u\in\{0,1\}^{n-t}}2^t-2\sum\limits_{\substack{u\in\{0,1\}^{n-t}\\w\in\{0,1\}^{t}}}f_w(u)\\
=&2^n-2\sum\limits_{x\in\{0,1\}^{n}}f(x)\\
=&2^n-2|\{x|f(x)=1\}|
=0.
\end{split}
\end{equation}

\end{proof}

In light of Theorem \ref{The2},  we have   Corollary \ref{Cor1} below, which provides a  sufficient  condition for determining whether a given Boolean function $f$ of  DJ problem is balanced by the absolute value of $\delta(u)$.


\begin{corollary}\label{Cor1} Suppose Boolean function $f:\{0,1\}^n \rightarrow \{0,1\}$, satisfies that it is either constant or balanced, and it is divided into $2^t$ subfunctions $f_w$ $(\forall u \in \{0,1\}^{n-t}, w \in \{0,1\}^{t}, f_w(u)=f(uw))$. If $\exists$ $u\in\{0,1\}^{n-t}$ such that $|\delta(u)|\neq 2^t$, then $f$ is balanced. 
\end{corollary}

By combining Theorem \ref{The4} with Theorem \ref{The2}, we can obtain Theorem 
\ref{The7} in the following. Also, we need some notations for convenience.

Suppose Boolean function $f:\{0,1\}^n \rightarrow \{0,1\}$, $f_{w'0}:\{0,1\}^{n-t} \rightarrow \{0,1\}$, $f_{w'1}:\{0,1\}^{n-t} \rightarrow \{0,1\}$, where $f_{w'0}(u)=f(uw'0)$, $f_{w'1}(u)=f(uw'1)$, $u\in\{0,1\}^{n-t}$, $w'\in\{0,1\}^{t-1}$. For all $u \in \{0,1\}^{n-t}$, denote
\begin{align}
E_{00}(u)=&|\{w'|f_{w'0}(u)=0,f_{w'1}(u)=0\}|.\\
E_{01}(u)=&|\{w'|f_{w'0}(u)=0,f_{w'1}(u)=1\}|.\\
E_{10}(u)=&|\{w'|f_{w'0}(u)=1,f_{w'1}(u)=0\}|.\\
E_{11}(u)=&|\{w'|f_{w'0}(u)=1,f_{w'1}(u)=1\}|.\\
\Delta(u)=&E_{00}(u)-E_{11}(u).\\
K(u)=&|\{w'|f_{w'0}(u)\oplus f_{w'1}(u)=1\}|.
\end{align}

It is clear that 
\begin{equation}\label{K=E_01+E_10}
K(u)=E_{01}(u)+E_{10}(u).
\end{equation}

According to the definition of $\Delta(u)$, for all $u \in \{0,1\}^{n-t}$, we have
\begin{equation}\label{Delta1}
\begin{split}
  \Delta(u)&=E_{00}(u)-E_{11}(u)\\
  &=\left[2^{t-1}-E_{11}(u)-K(u)\right]-E_{11}(u)\\
  &=2^{t-1}-K(u)-2E_{11}(u)\\
  &=2^{t-1}-\sum_{w'\in\{0,1\}^{t-1}}f_{w'0}(u)\oplus f_{w'1}(u)-2\sum_{w'\in\{0,1\}^{t-1}}f_{w'0}(u)\land f_{w'1}(u).
\end{split}
\end{equation}
and
\begin{equation}\label{Delta2}
\begin{split}
  \Delta(u)&=E_{00}(u)-E_{11}(u)\\
  &=E_{00}(u)-\left[2^{t-1}-E_{00}(u)-K(u)\right]\\
  &=-2^{t-1}+K(u)+2E_{00}(u)\\
  &=-2^{t-1}+\sum_{w'\in\{0,1\}^{t-1}}f_{w'0}(u)\oplus f_{w'1}(u)+2\sum_{w'\in\{0,1\}^{t-1}}(\lnot f_{w'0}(u))\land (\lnot f_{w'1}(u)).
\end{split}
\end{equation}

Theorem \ref{The7} below provides a sufficient and necessary condition for determining whether a given Boolean function $f$ of  DJ problem is constant or balanced in the light of $\Delta(u)$.

\begin{theorem}\label{The7} Suppose Boolean function $f:\{0,1\}^n \rightarrow \{0,1\}$, satisfies that it is either constant or balanced, and it is divided into $2^t$ subfunctions $f_{w'0}$ and  $f_{w'1}$ $(\forall u \in \{0,1\}^{n-t}, w' \in \{0,1\}^{t-1}, f_{w'0}(u)=f(uw'0), f_{w'1}(u)=f(uw'1))$. 
Then:
\begin{enumerate}[(1)]
 \item $f$ is constant  if and only if for all $u \in \{0,1\}^{n-t}$, $\Delta(u) = 2^{t-1}$ or for all $u \in \{0,1\}^{n-t}$, $\Delta(u) = -2^{t-1}$;
 \item $f$ is balanced  if and only if $\sum_{u\in\{0,1\}^{n-t}}\Delta(u)=0$.
 \end{enumerate}
\end{theorem}
\begin{proof}
Firstly, we prove that $f$ is constant if and only if for all $u \in \{0,1\}^{n-t}$, $\Delta(u) = 2^{t-1}$ or for all $u \in \{0,1\}^{n-t}$, $\Delta(u) = -2^{t-1}$.

(\romannumeral1) $\Longleftarrow$. If  for all $u \in \{0,1\}^{n-t}$, $\Delta(u) =2^{t-1}$, then according to equation (\ref{Delta1}), for all $u \in \{0,1\}^{n-t}$, we have 
\begin{equation}
\sum_{w'\in\{0,1\}^{t-1}}f_{w'0}(u)\oplus f_{w'1}(u)=0
\end{equation}
and
\begin{equation}
\sum_{w'\in\{0,1\}^{t-1}}f_{w'0}(u)\land f_{w'1}(u)=0. 
\end{equation}

So for all $x \in \{0,1\}^{n}$, $f(x)=0$, that is $f(x) \equiv 0$.

 If  for all $u \in \{0,1\}^{n-t}$, $\Delta(u) =-2^{t-1}$, then by equation (\ref{Delta2}), for all $u \in \{0,1\}^{n-t}$, we have 
\begin{equation}
\sum_{w'\in\{0,1\}^{t-1}}f_{w'0}(u)\oplus f_{w'1}(u)=0
\end{equation}
 and
 \begin{equation}
 \sum_{w'\in\{0,1\}^{t-1}}(\lnot f_{w'0}(u))\land (\lnot f_{w'1}(u))=0. 
  \end{equation}
  
 So for all $x \in \{0,1\}^{n}$, $f(x)=1$, that is $f(x) \equiv 1$.
 
 Therefore, if for all $u \in \{0,1\}^{n-t}$, $\Delta(u) =2^{t-1}$ or for all $u \in \{0,1\}^{n-t}$, $\Delta(u) =-2^{t-1}$, then $f$ is constant.

(\romannumeral2) $\Longrightarrow$. If $f$ is constant, then we have $f(x) \equiv 0$ or $f(x) \equiv 1$.

If $f(x) \equiv 0$, then for all $x \in \{0,1\}^{n}$, $f(x)=0$. So for all $u \in \{0,1\}^{n-t}$, we have 
 \begin{equation}
\sum_{w'\in\{0,1\}^{t-1}}f_{w'0}(u)\oplus f_{w'1}(u)=0
\end{equation}
and
\begin{equation}
\sum_{w'\in\{0,1\}^{t-1}}f_{w'0}(u)\land f_{w'1}(u)=0.
\end{equation}

With equation (\ref{Delta1}), for all $u \in \{0,1\}^{n-t}$, we have $\Delta(u) =2^{t-1}$.

If $f(x) \equiv 1$, then for all $x \in \{0,1\}^{n}$, $f(x)=1$. So for all $u \in \{0,1\}^{n-t}$, we have 
\begin{equation}
\sum_{w'\in\{0,1\}^{t-1}}f_{w'0}(u)\oplus f_{w'1}(u)=0
\end{equation}
and
\begin{equation}
\sum_{w'\in\{0,1\}^{t-1}}(\lnot f_{w'0}(u))\land (\lnot f_{w'1}(u))=0. 
\end{equation}

By equation (\ref{Delta2}), for all $u \in \{0,1\}^{n-t}$, we have $\Delta(u) =-2^{t-1}$.

 Therefore, if $f$ is constant, then for all $u \in \{0,1\}^{n-t}$, $\Delta(u) =2^{t-1}$ or for all $u \in \{0,1\}^{n-t}$, $\Delta(u) =-2^{t-1}$.

Secondly, we prove that $f$ is balanced if and only if $\sum_{u\in\{0,1\}^{n-t}}\Delta(u)=0$.

(\romannumeral3) $\Longleftarrow$. Suppose $\sum_{u\in\{0,1\}^{n-t}}\Delta(u)=0$, but $f$ is constant. 
From $f$ being  constant, it follows that $f\equiv 0$ or $f\equiv 1$, that is 
\begin{equation}
\sum\limits_{u\in\{0,1\}^{n-t}}E_{00}(u)=2^{n-1} 
\end{equation}
or 
\begin{equation}
\sum\limits_{u\in\{0,1\}^{n-t}}E_{11}(u)=2^{n-1}. 
\end{equation}

Therefore, we have
\begin{equation}
\sum\limits_{u\in\{0,1\}^{n-t}}\Delta(u)=2^{n-1}
\end{equation}
or
\begin{equation}
\sum\limits_{u\in\{0,1\}^{n-t}}\Delta(u)=-2^{n-1},
\end{equation}
which is contrary to the assumption $\sum_{u\in\{0,1\}^{n-t}}\Delta(u)=0$. 

As a result, if $\sum_{u\in\{0,1\}^{n-t}}\Delta(u)=0$, then $f$ is balanced.

(\romannumeral4) $\Longrightarrow$. 
If $f$ is balanced, 
then 
$|\{x|f(x)=0\}|=|\{x|f(x)=1\}|=2^{n-1}$. 
For all $u \in \{0,1\}^{n-t}$,  denote
\begin{equation}
F_0(u)=|\{w'|f_{w'0}(u)=0\}|.
\end{equation}

For all $u \in \{0,1\}^{n-t}$, denote
\begin{equation}
H_0(u)=|\{w'|f_{w'1}(u)=0\}|.
\end{equation}

Then we have
\begin{equation}\label{lemma1eq5}
\begin{split}
&\sum\limits_{u\in\{0,1\}^{n-t}}[F_0(u)+H_0(u)]\\
=&|\{x|f(x)=0\}|\\
=&2^{n-1}.
\end{split}
\end{equation}

Let
\begin{equation}\label{D}
D=\sum\limits_{u\in\{0,1\}^{n-t}}[E_{00}(u)+E_{11}(u)].
\end{equation}
\begin{equation}
T_0=\sum\limits_{u\in\{0,1\}^{n-t}}[F_0(u)-E_{01}(u)].
\end{equation}
\begin{equation}
T_1=\sum\limits_{u\in\{0,1\}^{n-t}}[H_0(u)-E_{10}(u)].
\end{equation}

Then
\begin{equation}\label{T_0=T_1}
T_0=T_1=\sum\limits_{u\in\{0,1\}^{n-t}}E_{00}(u).
\end{equation}
\begin{equation}\label{T_0+T_1}
\begin{split}
T_0+T_1=&\sum\limits_{u\in\{0,1\}^{n-t}}[F_0(u)+H_0(u)]-\sum\limits_{u\in\{0,1\}^{n-t}}[E_{01}(u)+E_{10}(u)]\\
&=2^{n-1}-(2^{n-1}-D)\\
&=D.
\end{split}
\end{equation}

Combining equation (\ref{T_0=T_1}) with equation (\ref{T_0+T_1}), we have
\begin{equation}\label{E_00=D/2}
\sum\limits_{u\in\{0,1\}^{n-t}}E_{00}(u)=\frac{D}{2}.
\end{equation}

Due to equation (\ref{D}) and equation (\ref{E_00=D/2}), we have
\begin{equation}\label{E_11=D/2}
\sum\limits_{u\in\{0,1\}^{n-t}}E_{11}(u)=\frac{D}{2}.
\end{equation}

From equation (\ref{E_00=D/2}) and equation (\ref{E_11=D/2}), we can deduce
\begin{equation}\label{Delta=0}
\begin{split}
\sum\limits_{u\in\{0,1\}^{n-t}}\Delta(u)=&\sum\limits_{u\in\{0,1\}^{n-t}}[E_{00}(u)-E_{11}(u)]\\
=&\sum\limits_{u\in\{0,1\}^{n-t}}E_{00}(u)-\sum\limits_{u\in\{0,1\}^{n-t}}E_{11}(u)\\
=&0.
\end{split}
\end{equation}

\end{proof}


In light of Theorem \ref{The7},  we have   Corollary \ref{Cor6} below, which provides a  sufficient  condition for determining whether a given Boolean function $f$ of  DJ problem is balanced by the absolute value of $\Delta(u)$.

\begin{corollary}\label{Cor6} Suppose Boolean function $f:\{0,1\}^n \rightarrow \{0,1\}$, satisfies that it is either constant or balanced.
 If $\exists$ $u\in\{0,1\}^{n-t}$ such that $|\Delta(u)|\neq 2^{t-1}$, then $f$ is balanced. 
\end{corollary}

\section{Design and analysis of Algorithm  \ref{algorithm3}} \label{sec:Distributed quantum algorithm for  DJ problem (Algorithm 1)}

  
  The  idea of our algorithms is  first based on the  essential characteristics of DJ problem in  distributed scenario, and is insights from the distributed Simon's quantum algorithm \cite{Tan2022DQCSimon}, by using   the specific unitary operator  and  quantum teleportation to combine the corresponding oracles of multiple subfunctions. Then, the joint information between multiple subfunctions is extracted by the specific controlled rotation operator based on the technique of HHL algorithm \cite{HHL_2009}, and the quantum state consistent with the structure of  DJ problem is obtained. 

\subsection{Design of  Algorithm \ref{algorithm3}}



In the following, we describe the unitary operators involved in Algorithm \ref{algorithm3}.

The operators $O_{f_0}$ and $O_{f_1}$ in  Algorithm \ref{algorithm3} are defined in equation (\ref{Of0}) and equation (\ref{Of1}), respectively. 
 
The operator $Z$ in Algorithm \ref{algorithm3} is the usual Pauli matrix $Z$.


\begin{figure}[H]
	\begin{minipage}{\linewidth}	
		\begin{algorithm}[H]	
			\caption{Distributed quantum algorithm for DJ problem (two distributed computing nodes)}
			\label{algorithm3}
			\begin{algorithmic}[1]
				\State $|\varphi_0\rangle = |0^{n-1}\rangle\ket{0}$;
				
				\State $|\varphi_1\rangle = \left(H^{\otimes n-1}\otimes I\right)|\psi_0\rangle$; 
				
				\State 
				$|\varphi_2\rangle=O_{f_0}|\varphi_1\rangle$;
				
				\State 
				$|\varphi_3\rangle=(I\otimes Z)|\varphi_2\rangle$;
				
				\State 
				$|\varphi_4\rangle=O_{f_1}|\varphi_3\rangle$;
				
				\State Measure the last qubit of $|\varphi_4\rangle$: if the result is 1, then output $f$ is balanced; if the result is 0, then denote the quantum state after measurement  as $\ket{\varphi_5}$;
				
				
				\State $|\varphi_6\rangle=\left(H^{\otimes n-1}\otimes I\right)|\varphi_5\rangle$;
				
				\State Measure the first $n-1$ qubits of $\ket{\varphi_6}$:  if the result is   $0^{n-1}$, then output $f$ is constant; if the result is not $0^{n-1}$, then output $f$ is balanced.
			\end{algorithmic}
		\end{algorithm}

		\end{minipage}
	\end{figure}

		\begin{figure}[H]
		\centering
		\includegraphics[width=6.3in]{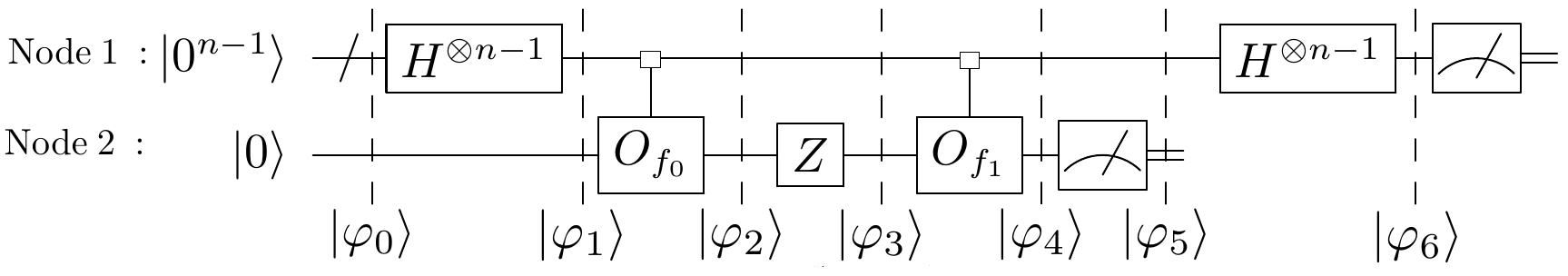}
		\caption{The circuit for the  distributed quantum algorithm for DJ problem (two computing nodes) (Algorithm \ref{algorithm3} ).}
		\label{algorithm_two_nodes (II)}
	\end{figure}


\subsection{The correctness analysis of  Algorithm \ref{algorithm3}}

In the following, we prove the correctness of  Algorithm \ref{algorithm3}. The state after the first step of Algorithm \ref{algorithm3} is:
\begin{align}
  |\varphi_1\rangle&=\frac{1}{\sqrt{2^{n-1}}}\sum_{u\in\{0,1\}^{n-1}}|u\rangle\ket{0}.
\end{align}

 Algorithm \ref{algorithm3} then queries the oracle $O_{f_{0}}$, resulting in the following state:
\begin{equation}
\begin{split}
  |\varphi_2\rangle&=O_{f_0}|\varphi_1\rangle\\
  &=\frac{1}{\sqrt{2^{n-1}}}\sum_{u\in\{0,1\}^{n-1}}|u\rangle|f_{0}(u)\rangle\\ 
\end{split}
\end{equation}

Then, the quantum gate $Z$ on  $|\varphi_2\rangle$ is applied to get the following states:
\begin{equation}
\begin{split}
  |\varphi_3\rangle=&(I\otimes Z)|\varphi_2\rangle\\
  =&\frac{1}{\sqrt{2^{n-1}}}\sum_{u\in\{0,1\}^{n-1}}(-1)^{f_0(u)}|u\rangle|f_0(u)\rangle.
\end{split}
\end{equation}

After applying the operator $O_{f_1}$ on $|\varphi_3\rangle$, we have the following state:
\begin{equation}
\begin{split}
  |\varphi_4\rangle=&O_{f_1}|\varphi_3\rangle\\
  =&\frac{1}{\sqrt{2^{n-1}}}\sum_{u\in\{0,1\}^{n-1}}(-1)^{f_0(u)}|u\rangle|f_0(u)\oplus f_1(u)\rangle\\
  =&\frac{1}{\sqrt{2^{n-1}}}\sum\limits_{\substack{u\in\{0,1\}^{n-1}\\ f_0(u)\oplus f_1(u)=0}}(-1)^{f_0(u)}|u\rangle\ket{0}+\frac{1}{\sqrt{2^{n-1}}}\sum\limits_{\substack{u\in\{0,1\}^{n-1}\\ f_0(u)\oplus f_1(u)=1}}(-1)^{f_0(u)}|u\rangle\ket{1}.
\end{split}
\end{equation}

After measuring on the last qubit of $|\varphi_4\rangle$, if the result is $1$, then there $\exists$ $u\in\{0,1\}^{n-1}$ such that $f_0(u)\oplus f_1(u)=1$. From Corollary \ref{Cor3}, we know that $f$ is balanced. If the result is $0$, then we get the state:
\begin{equation}
\begin{split}
|\varphi_5\rangle=&\frac{1}{\sqrt{M}}\sum\limits_{\substack{u\in\{0,1\}^{n-1}\\ f_0(u)\oplus f_1(u)=0}}(-1)^{f_0(u)}|u\rangle\ket{0}, \\
\end{split}
\end{equation}
where $M=|\{u\in\{0,1\}^{n-1}|f_0(u)\oplus f_1(u)=0\}|$.


After Hadamard transformation on the first $n-1$ qubits of $|\varphi_5\rangle$, we obtain the following state:
\begin{equation}
\begin{split}
	|\varphi_6\rangle=&\left(H^{\otimes n-1}\otimes I\right)|\varphi_5\rangle\\
	=&\frac{1}{\sqrt{2^{n-1}M}}\sum\limits_{\substack{u,z\in\{0,1\}^{n-1}\\ f_0(u)\oplus f_1(u)=0}}(-1)^{f_0(u)+u\cdot z}\ket{z}\ket{0}.
\end{split}	
\end{equation}


The probability of measuring the first $n-1$ qubits of  $|\varphi_6\rangle$ with the result of $0^{n-1}$ is
\begin{equation}
\left|\frac{1}{\sqrt{2^{n-1}M}}\sum\limits_{\substack{u\in\{0,1\}^{n-1}\\ f_0(u)\oplus f_1(u)=0}}(-1)^{f_0(u)}\right|^2.
\end{equation}

After measuring on the first $n-1$ qubits of  $|\varphi_6\rangle$, according to Theorem \ref{The4}, 
if the result is  $0^{n-1}$, then $f$ is  constant, otherwise $f$ is balanced.

\subsection{Comparison with other algorithms}


First, we compare  with the   distributed quantum algorithm for DJ problem proposed previously \cite{avron_quantum_2021}. 
For $t=1$, the algorithm in paper \cite{avron_quantum_2021} has certain error, but Algorithm \ref{algorithm3} can  solve it exactly.

Second, we compare  with distributed classical deterministic algorithm. Algorithm \ref{algorithm3} needs one query for each oracle  to solve DJ problem. However, distributed classical deterministic algorithm needs to query oracles $O(2^{n-1})$ times in the worst case. Therefore, Algorithm \ref{algorithm3} has the advantage of exponential acceleration compared with the   distributed classical deterministic algorithm.

Third, we compare with DJ algorithm. In DJ algorithm, the number of qubits required by the implementation circuit of oracle corresponding to Boolean function $f$ is $n+1$. In Algorithm \ref{algorithm3}, the number of qubits required by the implementation circuit of oracle corresponding to subfunctions $f_0$ and $f_1$ is $n$. 






\begin{remark}
However,  Algorithm \ref{algorithm3}  cannot exactly solve  DJ problem in  distributed scenario  with multiple computing nodes. In \ref{Distributed quantum algorithm for DJ problem with errors (four distributed computing nodes)}, we  give an example to demonstrate that the algorithm cannot exactly solve  DJ problem    for the case with four computing nodes.

\end{remark}

\section{Design and analysis of Algorithm  \ref{algorithm2}} \label{sec:Distributed quantum algorithm for  DJ problem (Algorithm 2)}




Since Algorithm  \ref{algorithm3} cannot exactly solve  DJ problem in  distributed scenario with multiple computing nodes, we use the new structural features of DJ problem  to design Algorithm \ref{algorithm2}. 


\subsection{Design of Algorithm \ref{algorithm2}}

Below we introduce related notation and  function that are used in Algorithm \ref{algorithm2}.

Let $[N]$ represent the set of integers $\{0,1,\cdots, 2^t-1\}$, ane let $BI:\{0,1\}^t \rightarrow [N]$ be the function to convert a binary string of $t$ bits to an equal decimal integer.

Operator $O'_{f_w}$ in Algorithm \ref{algorithm2}
 is defined as: 
\begin{equation}
O'_{f_w}\ket{u}\ket{b}\ket{c}=\ket{u}\ket{b}\ket{c\oplus f_w(u)},
\end{equation}
where $w\in \{0,1\}^t$, $u\in\{0,1\}^{n-t}$, $b\in\{0,1\}^{BI(w)}$ and $c\in \{0,1\}$.

Operator $U$ in Algorithm \ref{algorithm2} 
is defined as: 
\begin{equation}\label{U}
U\left(\bigotimes_{i\in\{0,1\}^{t}}\ket{a_i}\right)|b\rangle
=\left(\bigotimes_{i\in\{0,1\}^{t}}\ket{a_i}\right)\Ket{b\oplus \left(2^{t}-2\sum\limits_{i\in\{0,1\}^{t}}a_i\right)},
\end{equation}
where  $a_i\in \{0,1\}$, $i\in \{0,1\}^t$ and $b\in \{0,1\}^{t+2}$.
In fact,  $U$ is  unitary.


Another  operator $R$ in Algorithm \ref{algorithm2} 
is defined as: 
\begin{equation}\label{R}
	R\ket{d}\ket{e}=\ket{d}\left(\frac{d}{2^t}\ket{e}+(-1)^{e}\cdot\sqrt{1-\left(\frac{d}{2^t}\right)^2}\Ket{1\oplus e}\right),
\end{equation}
where  $ d \in \{0,1\}^{t+2}$ and $e\in\{0,1\}$.
Also,  $R$ is unitary.

\begin{lemma}\label{Lem_U} 
The  operator $U$  in equation (\ref{U})  and  the  operator $R$  in equation (\ref{R}) are both unitary.
\end{lemma}
\begin{proof}
Suppose $\left(\bigotimes_{i\in\{0,1\}^{t}}\ket{a_i}\right)|b\rangle\neq  \left(\bigotimes_{i\in\{0,1\}^{t}}\ket{a'_i}\right)|b'\rangle$. Then 
\begin{equation}
\begin{split}
&\bra{b'}\left(\bigotimes_{i\in\{0,1\}^{t}}\bra{a'_i}\right)U^{\dagger}U\left(\bigotimes_{i\in\{0,1\}^{t}}\ket{a_i}\right)|b\rangle\\
=&\Bra{b'\oplus \left(2^{t}-2\sum\limits_{i\in\{0,1\}^{t}}a'_i\right)}\left(\bigotimes_{i\in\{0,1\}^{t}}\bra{a'_i}\right)\left(\bigotimes_{i\in\{0,1\}^{t}}\ket{a_i}\right)\Ket{b\oplus \left(2^{t}-2\sum\limits_{i\in\{0,1\}^{t}}a_i\right)}\\
=&0.
\end{split}
\end{equation}

Therefore, $U$ is unitary.

Suppose $\ket{d}\ket{e}\neq\ket{d'}\ket{e'}$. Then 
\begin{equation}
\begin{split}
&\bra{e'}\bra{d'}R^{\dagger}R\ket{d}\ket{e}\\
=&\left(\frac{d'}{2^t}\bra{e'}+(-1)^{e'}\cdot\sqrt{1-\left(\frac{d'}{2^t}\right)^2}\bra{1\oplus e'}\right)\braket{d'|d}\left(\frac{d}{2^t}\ket{e}+(-1)^{e}\cdot\sqrt{1-\left(\frac{d}{2^t}\right)^2}\Ket{1\oplus e}\right)\\
=&0.
\end{split}
\end{equation}

Therefore, $R$  is unitary.
\end{proof}

\begin{figure}[H]
  \centering
  \includegraphics[width=6.4in]{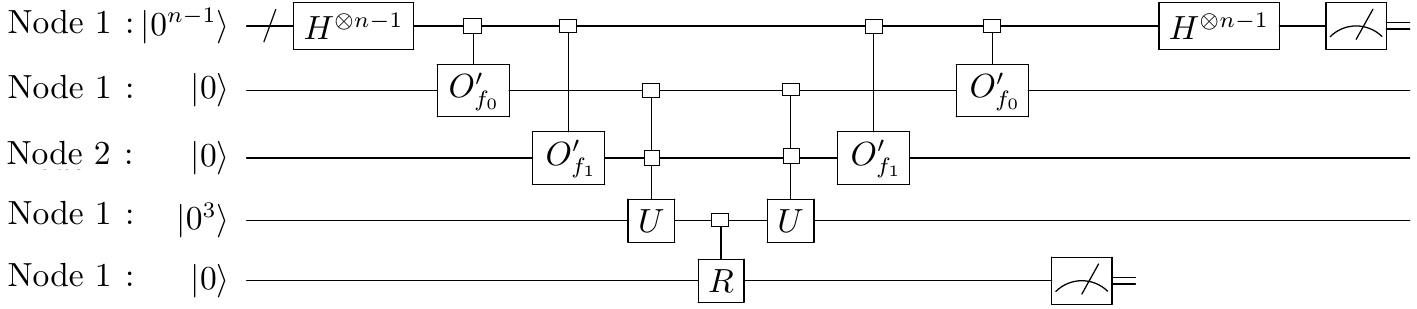}
  \caption{The circuit for the distributed quantum algorithm for DJ problem  (two computing nodes) (Algorithm \ref{algorithm2} ).}
  \label{algorithm_multiple_nodes (I)}
\end{figure}

\begin{figure}[H]
  \begin{minipage}{\linewidth}	
    \begin{algorithm}[H]
      \caption{Distributed quantum algorithm for DJ problem ($2^t$  distributed computing nodes)}
      \label{algorithm2}
      \begin{algorithmic}[1]
        
        \State $\ket{\phi_0} = |0^{n-t}\rangle\ket{0^{2^{t}+t+3}}$;

         \State $|\phi_1\rangle=\left(H^{\otimes n-t}\otimes I^{\otimes {2^{t}+t+3}}\right)\ket{\phi_0}$;

        \State $\ket{\phi_2}=\prod\limits_{w\in\{0,1\}^t}\left(O'_{f_{w}}\otimes I^{\otimes 2^t+t+2-BI(w)}\right)\ket{\phi_1}$;
        
        

         \State $\ket{\phi_3}=\left(I^{\otimes {n-t}}\otimes U\otimes I\right)\ket{\phi_2}$;
        
        \State $\ket{\phi_4}=\left(I^{\otimes {n-t+2^t}}\otimes R\right)\ket{\phi_3}$;
        
        \State $\ket{\phi_5}=\left(\prod\limits_{w\in\{0,1\}^t}\left(O'_{f_{w}}\otimes I^{\otimes 2^t+t+2-BI(w)}\right)\right)\left(I^{\otimes {n-t}}\otimes U\otimes I\right)\ket{\phi_4}$;

      \State  Measure the last qubit of  $\ket{\phi_5}$: if the result is 1, then output $f$ is balanced; if the result is 0, then denote the quantum state after measurement  as $\ket{\phi_6}$;

        \State $|\phi_{7}\rangle=
        \left(H^{\otimes n-t}\otimes I^{\otimes {2^t+t+3}}\right)|\phi_{6}\rangle$;

        \State Measure the first $n-t$ qubits of $\ket{\phi_7}$: if the result is  $0^{n-t}$, then output $f$ is constant;  if the result is not $0^{n-t}$, then output $f$ is balanced. 
      \end{algorithmic}
    \end{algorithm}

  \end{minipage}
  \end{figure}


\begin{figure}[H]
  \centering
  \includegraphics[width=6.4in]{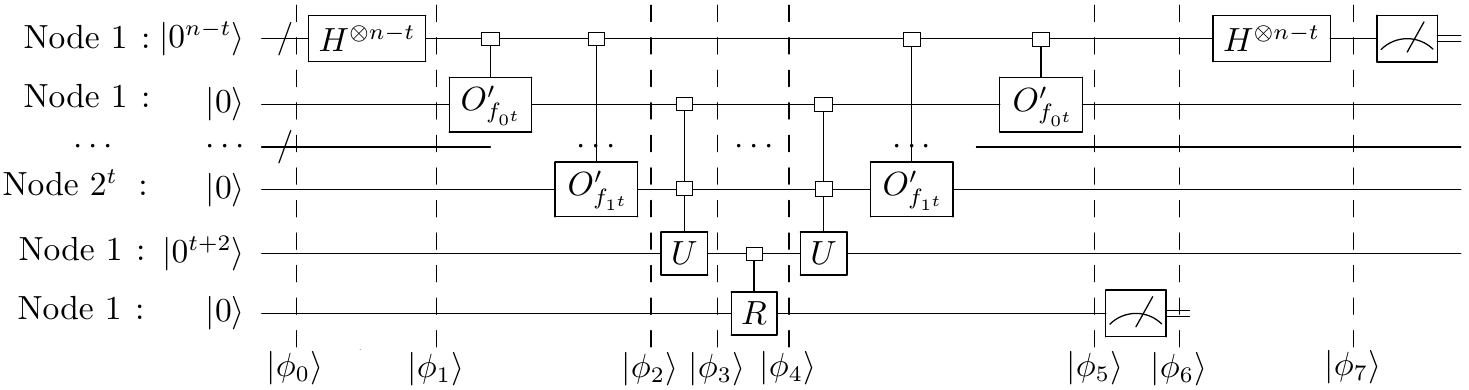}
  \caption{The circuit for the distributed quantum algorithm for DJ problem  ($2^t$ computing nodes) (Algorithm \ref{algorithm2} ).}
  \label{algorithm_multiple_nodes (I)}
\end{figure}





\subsection{Correctness analysis of  Algorithm \ref{algorithm2}}


In the following, we prove the correctness of  Algorithm \ref{algorithm2}.
The state after the first step of  Algorithm \ref{algorithm2} is: 
\begin{equation}
\begin{split}
  |\phi_1\rangle=&\left(H^{\otimes n-t}\otimes I^{\otimes {2^{t}+t+3}}\right)\ket{\phi_0}\\
  =&\frac{1}{\sqrt{2^{n-t}}}\sum_{u\in\{0,1\}^{n-t}}|u\rangle\ket{0^{2^{t}+t+3}}.
\end{split}
\end{equation}

Then Algorithm \ref{algorithm2}  queries the oracles $O'_{f_{w}}$ $\left(w\in\{0,1\}^t\right)$,  resulting in the following state:
\begin{equation}
\begin{split}
  |\phi_2\rangle=&\prod\limits_{w\in\{0,1\}^t}\left(O'_{f_{w}}\otimes I^{\otimes 2^t+t+2-BI(w)}\right)\ket{\phi_1}\\
  =&\frac{1}{\sqrt{2^{n-t}}}\sum_{u\in\{0,1\}^{n-t}}|u\rangle\bigotimes\limits_{w\in\{0,1\}^t}\ket{f_w(u)}
  |0^{t+3}\rangle.
\end{split}
\end{equation}

After  applying  operator $U$, we have the following state:
\begin{equation}
\begin{split}
  |\phi_3\rangle=&\left(I^{\otimes {n-t}}\otimes U\otimes I\right)\ket{\phi_2}\\=&\frac{1}{\sqrt{2^{n-t}}}\sum_{u\in\{0,1\}^{n-t}}|u\rangle
  \bigotimes\limits_{w\in\{0,1\}^t}\ket{f_w(u)}\Ket{\delta(u)}\Ket{0}, 
\end{split}
\end{equation}
where $\delta(u)=2^t-2\sum_{w\in\{0,1\}^t}f_w(u)$.

After  applying operator $R$, we have the following state:
\begin{equation}
\begin{split}
  |\phi_4\rangle=&\left(I^{\otimes {n-t+2^t}}\otimes R\right)\ket{\phi_3}\\
=&\frac{1}{\sqrt{2^{n-t}}}\sum_{u\in\{0,1\}^{n-t}}|u\rangle
\bigotimes\limits_{w\in\{0,1\}^t}\ket{f_w(u)}\Ket{\delta(u)}
\left(\frac{\delta(u)}{2^t}\ket{0}+\sqrt{1-\left(\frac{\delta(u)}{2^t}\right)^2}\ket{1}\right).
\end{split}
\end{equation}

After  applying  operator $U$ and $O'_{f_{w}}$ $(w\in\{0,1\}^t)$, resulting in the following state:
\begin{equation}
\begin{split}
\ket{\phi_5}=&\left(\prod\limits_{w\in\{0,1\}^t}\left(O'_{f_{w}}\otimes I^{\otimes 2^t+t+2-BI(w)}\right)\right)\left(I^{\otimes {n-t}}\otimes U\otimes I\right)\ket{\phi_4}\\=&\frac{1}{\sqrt{2^{n-t}}}\sum_{u\in\{0,1\}^{n-t}}|u\rangle|0^{2^t+t+2}\rangle\left(\frac{\delta(u)}{2^t}\ket{0}+\sqrt{1-\left(\frac{\delta(u)}{2^t}\right)^2}\ket{1}\right).
\end{split}
\end{equation}


After measuring  the last qubit of $\ket{\phi_5}$, if the result is $1$, then $\exists$ $u\in\{0,1\}^{n-t}$ such that $|\delta(u)|\neq 2^t$. From Corollary \ref{Cor1}, we know that $f$ is balanced. If the result is $0$, then we get the state:
\begin{equation}
\begin{split}
&|\phi_6\rangle=\sum_{u\in\{0,1\}^{n-t}}\frac{\delta(u)}{\sqrt{\sum\limits_{u\in\{0,1\}^{n-t}}\delta^2(u)}}|u\rangle\ket{0^{2^t+t+3}}.
\end{split}
\end{equation}


With Hadamard transformation on  the first $n-t$ qubits of $\ket{\phi_6}$, we get the following state:
\begin{equation}
\begin{split}
	|\phi_{7}\rangle=&
        \left(H^{\otimes n-t}\otimes I^{\otimes {2^t+t+3}}\right)|\phi_{6}\rangle\\
	=&\sum_{u,z\in\{0,1\}^{n-t}}\frac{\delta(u)}{\sqrt{\sum\limits_{u\in\{0,1\}^{n-t}}\delta^2(u)}}\frac{(-1)^{u\cdot z}}{\sqrt{2^{n-t}}}\ket{z}\ket{0^{2^t+t+3}}.
\end{split}
\end{equation} 



The probability of measuring the first $n-t$ qubits of $\ket{\phi_7}$ with the result of $0^{n-t}$ is
\begin{equation}
\left|\frac{1}{\sqrt{2^{n-t}\sum\limits_{u\in\{0,1\}^{n-t}}\delta^2(u)}}\sum\limits_{u\in\{0,1\}^{n-t}}\delta(u)\right|^2.
\end{equation}

After measuring the first $n-t$ qubits of $\ket{\phi_7}$, according to Theorem \ref{The2}, if the result is  $0^{n-t}$, then $f$ is  constant, otherwise $f$ is balanced. 




\subsection{Comparison with other algorithms}



We compare  with distributed classical deterministic algorithm. Algorithm \ref{algorithm2} needs two queries for each oracle  to solve DJ problem. However, distributed classical deterministic  algorithm needs to query oracles $O(2^{n-t})$ times in the worst case. Therefore, Algorithm \ref{algorithm2} has the advantage of exponential acceleration compared with the distributed classical deterministic algorithm.

Then, we compare  with DJ algorithm. In DJ algorithm, the number of qubits required by the implementation circuit of oracle corresponding to Boolean function $f$ is $n+1$. In Algorithm \ref{algorithm2}, the number of qubits required by the implementation circuit of oracle corresponding to subfunction $f_w$ is $n-t+1$. 

\section{Design and analysis of Algorithm  \ref{algorithm5}} \label{sec:Distributed quantum algorithm for  DJ problem (Algorithm 3)}





We synthesised the structural features of  DJ problem in  distributed scenario used in Algorithm \ref{algorithm3} and Algorithm \ref{algorithm2}  to design Algorithm \ref{algorithm5}. Algorithm \ref{algorithm5} combines the design methods and advantages of Algorithm \ref{algorithm3} and Algorithm \ref{algorithm2}, but the number of qubits required to implement certain unitary operators in Algorithm \ref{algorithm5} is less than that in Algorithm \ref{algorithm2}.

\subsection{Design of Algorithm \ref{algorithm5}}

In the following, we describe the unitary operators involved in Algorithm \ref{algorithm5}. 




Operator $O''_{f_{w'0}}$ in Algorithm \ref{algorithm5}
 is defined as: 
\begin{equation}
O''_{f_{w'0}}\ket{u}\ket{b}\ket{c}=\ket{u}\ket{b}\ket{c\oplus f_{w'0}(u)},
\end{equation}
where $w'\in \{0,1\}^{t-1}$, $u\in\{0,1\}^{n-t}$, $b\in\{0,1\}^{3\cdot BI(w'0)/2}$ and $c\in \{0,1\}$.

Operator $O''_{f_{w'1}}$ in Algorithm \ref{algorithm5}
 is defined as: 
\begin{equation}
O''_{f_{w'1}}\ket{u}\ket{b}\ket{c}=\ket{u}\ket{b}\ket{c\oplus f_{w'1}(u)},
\end{equation}
where $w'\in \{0,1\}^{t-1}$, $u\in\{0,1\}^{n-t}$, $b\in\{0,1\}^{3\cdot BI(w'0)/2+1}$ and $c\in \{0,1\}$. The function $BI$ is defined in Algorithm \ref{algorithm2}.

Operator $A$ in  Algorithm \ref{algorithm5} is defined as: 
\begin{equation}\label{A}
A\left(\bigotimes_{i\in\{0,1\}^{t-1}\setminus\{1^{t-1}\}}\ket{a_i}\ket{b}\right)\ket{a_{1^{t-1}}}\ket{c}|d\rangle\\
=\left(\bigotimes_{i\in\{0,1\}^{t-1}\setminus\{1^{t-1}\}}\ket{a_i}\ket{b}\right)\ket{a_{1^{t-1}}}\ket{c}\Ket{d\oplus\sum_{i\in\{0,1\}^{t-1}}a_i},
\end{equation}
 where $a_i\in \{0,1\}$, $i\in \{0,1\}^{t-1}$, $b\in\{0,1\}^2$, $c\in \{0,1\}$  and $d\in \{0,1\}^{t}$.




Operator $V$ in  Algorithm \ref{algorithm5} is defined as: 
\begin{equation}\label{V}
\begin{split}
V\ket{a}\ket{b}
\ket{c}
=&\ket{a}\ket{b}
\Ket{c\oplus\left(2^{t-1}-a-2b\right)},
\end{split}
\end{equation}
where $ a,b \in \{0,1\}^t$ 
and $c \in \{0,1\}^{t+1}$.




Operator $R'$ in  Algorithm \ref{algorithm5} is defined as: 
\begin{equation}\label{R'}
	R'\ket{d}\ket{e}=\ket{d}\left(\frac{d}{2^{t-1}}\ket{e}+(-1)^{e}\cdot\sqrt{1-\left(\frac{d}{2^{t-1}}\right)^2}\Ket{1\oplus e}\right),
\end{equation}
where $ d \in \{0,1\}^{t+1}$ and $e\in\{0,1\}$.

\begin{lemma}\label{Lem_AVR} 
The operator $A$ in equation (\ref{A}), the operator $V$ in equation (\ref{V}) and the operator $R'$ in equation (\ref{R'}) are all unitary.
\end{lemma}
\begin{proof}
Similar to the proof of Lemma \ref{Lem_U}. 
\end{proof}

\begin{figure}[H]
  \centering
  \includegraphics[width=6.4in]{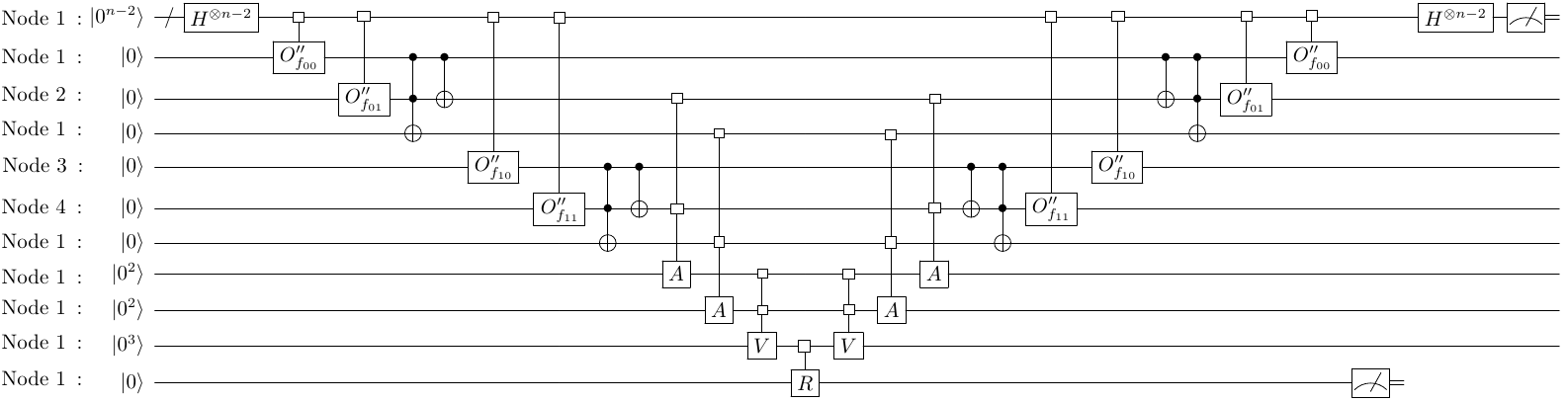}
  \caption{The circuit for the distributed quantum algorithm for DJ problem  (four computing nodes) (Algorithm \ref{algorithm5} ).}
  \label{algorithm_multiple_nodes (III)}
\end{figure}

\begin{figure}[H]
  \begin{minipage}{\linewidth}
    \begin{algorithm}[H]
      \caption{Distributed quantum algorithm for DJ problem ($2^t$  distributed computing nodes)}
      \label{algorithm5}
      \begin{algorithmic}[1]
        \State $|\Phi_0\rangle = |0^{n-t}\rangle|0^{3\cdot 2^{t-1}+3t+2}\rangle$;

        \State $\ket{\Phi_1} =\left(H^{\otimes n-t}\otimes I^{\otimes {3\cdot 2^{t-1}+3t+2}}\right)|\Phi_0\rangle$;

        \State
        $|\Phi_2\rangle=
        \prod\limits_{w'\in\{0,1\}^{t-1}}\left[\left(I^{\otimes n-t+3\cdot BI(w'0)/2}\otimes{\rm CNOT}\otimes I^{\otimes 3\cdot2^{t-1}+3t-3\cdot BI(w'0)/2}\right)\right.$
        
       $\quad\left(I^{\otimes n-t+3\cdot BI(w'0)/2}\otimes{\rm CCNOT}\otimes I^{\otimes 3\cdot2^{t-1}+3t-3\cdot BI(w'0)/2-1}\right)$
       
       $\quad\left.\left(O''_{f_{w'1}}\otimes I^{\otimes 3\cdot2^{t-1}+3t-3\cdot BI(w'0)/2}\right)\left(O''_{f_{w'0}}\otimes I^{\otimes 3\cdot2^{t-1}+3t-3\cdot BI(w'0)/2+1}\right)\right]\ket{\Phi_1}$

         \State 
         $|\Phi_3\rangle=\prod\limits_{i=1}^2\left(I^{\otimes n-t+i}\otimes A\otimes I^{\otimes (3-i)t+2}\right)\ket{\Phi_2}$;

         \State $\Ket{\Phi_4}=\left(I^{\otimes n-t+3\cdot 2^{t-1}}\otimes V\otimes I\right)\Ket{\Phi_3}$;
         
         
         

        \State$\Ket{\Phi_{5}}=\left(I^{\otimes n+3\cdot2^{t-1}+t}\otimes R'\right)\Ket{\Phi_4}$;
        
        
         

       \State 
       $\ket{\Phi_6}=\prod\limits_{w'\in\{0,1\}^{t-1}}\left[\left(I^{\otimes n-t+3\cdot BI(w'0)/2}\otimes{\rm CNOT}\otimes I^{\otimes 3\cdot2^{t-1}+3t-3\cdot BI(w'0)/2}\right)\right.$
        
       $\quad\left(I^{\otimes n-t+3\cdot BI(w'0)/2}\otimes{\rm CCNOT}\otimes I^{\otimes 3\cdot2^{t-1}+3t-3\cdot BI(w'0)/2-1}\right)$
       
       $\quad\left.\left(O''_{f_{w'1}}\otimes I^{\otimes 3\cdot2^{t-1}+3t-3\cdot BI(w'0)/2}\right)\left(O''_{f_{w'0}}\otimes I^{\otimes 3\cdot2^{t-1}+3t-3\cdot BI(w'0)/2+1}\right)\right]$
       
       $\quad\prod\limits_{i=1}^2\left(I^{\otimes n-t+i}\otimes A\otimes I^{\otimes (3-i)t+2}\right)\left(I^{\otimes n-t+3\cdot 2^{t-1}}\otimes V\otimes I\right)\ket{\Phi_5}$;

       \State  Measure the last qubit of  $\ket{\Phi_6}$: if the result is 1, then output $f$ is balanced; if the result is 0, then denote the quantum state after measurement  as $\ket{\Phi_7}$;
       
       

        \State $|\Phi_{8}\rangle=\left(H^{\otimes n-t}\otimes I^{\otimes {3\cdot 2^{t-1}+3t+2}}\right)|\Phi_{7}\rangle$;

        \State Measure the first $n-t$ qubits of $\ket{\Phi_8}$: if the result is  $0^{n-t}$, then output $f$ is constant;  if the result is not $0^{n-t}$, then output $f$ is balanced.
      \end{algorithmic}
    \end{algorithm}

  \end{minipage}
  \end{figure}

\begin{figure}[H]
  \centering
  \includegraphics[width=6.5in]{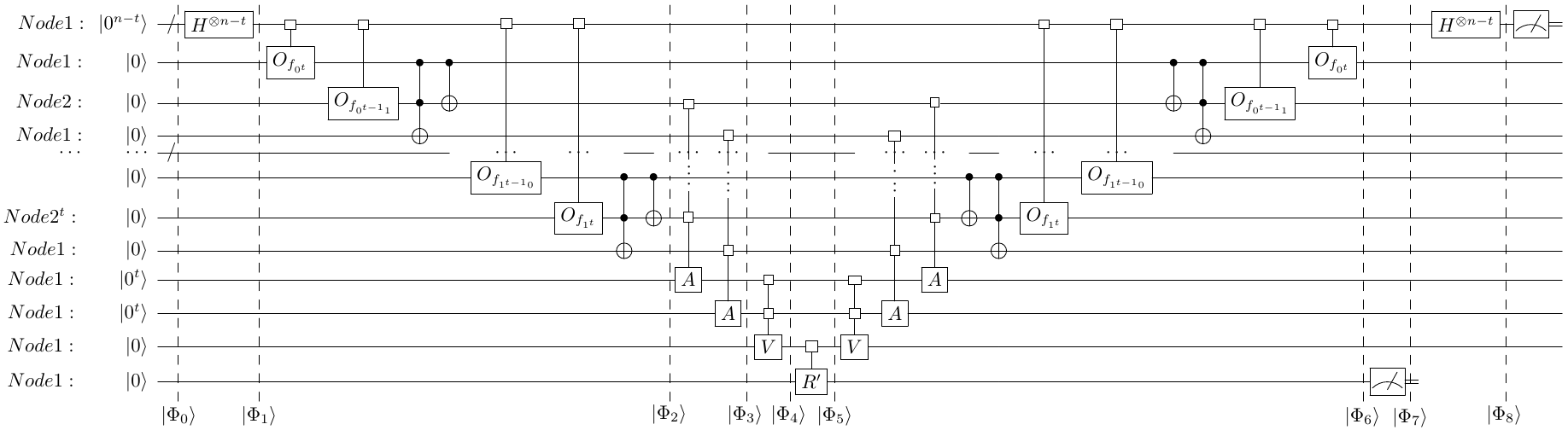}
  \caption{The circuit for the distributed quantum algorithm for DJ problem  ($2^t$ computing nodes) (Algorithm \ref{algorithm5} ).}
  \label{algorithm_multiple_nodes (III)}
\end{figure}

\subsection{Correctness analysis of  Algorithm \ref{algorithm5}}

In the following, we prove the correctness of  Algorithm \ref{algorithm5}.
The state after the first step of  Algorithm \ref{algorithm5} is:
\begin{equation}
\begin{split}
  \ket{\Phi_1}=&\left(H^{\otimes n-t}\otimes I^{\otimes {3\cdot 2^{t-1}+3t+2}}\right)|\Phi_0\rangle\\
  =&\frac{1}{\sqrt{2^{n-t}}}\sum_{u\in\{0,1\}^{n-t}}|u\rangle\ket{0^{3\cdot 2^{t-1}+3t+2}}.
\end{split} 
\end{equation} 

Then Algorithm \ref{algorithm5} queries the oracles $O''_{f_{w'0}}$ and $O''_{f_{w'1}}$ $(w'\in\{0,1\}^{t-1})$, and applies the operators ${\rm CCNOT}$ and ${\rm CNOT}$, resulting in the following state:
\begin{equation}
\begin{split}
  |\Phi_2\rangle=&\prod\limits_{w'\in\{0,1\}^{t-1}}\left[\left(I^{\otimes n-t+3\cdot BI(w'0)/2}\otimes{\rm CNOT}\otimes I^{\otimes 3\cdot2^{t-1}+3t-3\cdot BI(w'0)/2}\right)\right.\\
        &\left(I^{\otimes n-t+3\cdot BI(w'0)/2}\otimes{\rm CCNOT}\otimes I^{\otimes 3\cdot2^{t-1}+3t-3\cdot BI(w'0)/2-1}\right)\\
       &\left.\left(O''_{f_{w'1}}\otimes I^{\otimes 3\cdot2^{t-1}+3t-3\cdot BI(w'0)/2}\right)\left(O''_{f_{w'0}}\otimes I^{\otimes 3\cdot2^{t-1}+3t-3\cdot BI(w'0)/2+1}\right)\right]\ket{\Phi_1}\\
=&\frac{1}{\sqrt{2^{n-t}}}\sum_{u\in\{0,1\}^{n-t}}|u\rangle\left(\bigotimes_{w'\in\{0,1\}^{t-1}}|f_{w'0}(u)\rangle|f_{w'0}(u)\oplus f_{w'1}(u)\rangle|f_{w'0}(u)\land f_{w'1}(u)\rangle\right)\ket{0^{3t+2}}.
\end{split} 
\end{equation}

By applying the operator $A$ on $\ket{\Phi_2}$, we obtain the following state:
\begin{equation}
\begin{split}
|\Phi_3\rangle=&\prod\limits_{i=1}^2\left(I^{\otimes n-t+i}\otimes A\otimes I^{\otimes (3-i)t+2}\right)\ket{\Phi_2}\\
=&\frac{1}{\sqrt{2^{n-t}}}\sum_{u\in\{0,1\}^{n-t}}|u\rangle\left(\bigotimes_{w'\in\{0,1\}^{t-1}}|f_{w'0}(u)\rangle|f_{w'0}(u)\oplus f_{w'1}(u)\rangle|f_{w'0}(u)\land f_{w'1}(u)\rangle\right)\\
&\left|\sum_{w'\in\{0,1\}^{t-1}}f_{w'0}(u)\oplus f_{w'1}(u)\right\rangle\Ket{\sum_{w'\in\{0,1\}^{t-1}}f_{w'0}(u)\land f_{w'1}(u)}\Ket{0^{t+2}}.
\end{split} 
\end{equation}    

By applying the operator $V$ on $\ket{\Phi_3}$, we get the following state:
\begin{equation}
\begin{split}
\Ket{\Phi_4}=&\left(I^{\otimes n-t+3\cdot 2^{t-1}}\otimes V\otimes I\right)\Ket{\Phi_3}\\
=&\frac{1}{\sqrt{2^{n-t}}}\sum_{u\in\{0,1\}^{n-t}}|u\rangle\left(\bigotimes_{w'\in\{0,1\}^{t-1}}|f_{w'0}(u)\rangle|f_{w'0}(u)\oplus f_{w'1}(u)\rangle|f_{w'0}(u)\land f_{w'1}(u)\rangle\right)\\
&\left|\sum_{w'\in\{0,1\}^{t-1}}f_{w'0}(u)\oplus f_{w'1}(u)\right\rangle\Ket{\sum_{w'\in\{0,1\}^{t-1}}f_{w'0}(u)\land f_{w'1}(u)}
\Ket{\Delta(u)}\Ket{0},
\end{split} 
\end{equation}   
where $\Delta(u)=2^{t-1}-\sum_{w'\in\{0,1\}^{t-1}}f_{w'0}(u)\oplus f_{w'1}(u)-2\sum_{w'\in\{0,1\}^{t-1}}f_{w'0}(u)\land f_{w'1}(u)$.

Then we apply the operator  $R'$ on $\Ket{\Phi_4}$,  and the following state is obtained:
\begin{equation}
\begin{split}
\Ket{\Phi_{5}}=&\left(I^{\otimes n+3\cdot2^{t-1}+t}\otimes R'\right)\Ket{\Phi_4}\\
=&\frac{1}{\sqrt{2^{n-t}}}\sum_{u\in\{0,1\}^{n-t}}|u\rangle\left(\bigotimes_{w'\in\{0,1\}^{t-1}}|f_{w'0}(u)\rangle|f_{w'0}(u)\oplus f_{w'1}(u)\rangle|f_{w'0}(u)\land f_{w'1}(u)\rangle\right)\\
&\left|\sum_{w'\in\{0,1\}^{t-1}}f_{w'0}(u)\oplus f_{w'1}(u)\right\rangle\Ket{\sum_{w'\in\{0,1\}^{t-1}}f_{w'0}(u)\land f_{w'1}(u)}
\Ket{\Delta(u)}\\
&\left(\frac{\Delta(u)}{2^{t-1}}\ket{0}+\sqrt{1-\left(\frac{\Delta(u)}{2^{t-1}}\right)^2}\ket{1}\right).
\end{split} 
\end{equation}  

In addition, we restore  the middle $3\cdot 2^{t-1}+3t+1$ qubits of $\Ket{\Phi_{5}}$ to $\Ket{0^{3\cdot 2^{t-1}+3t+1}}$, resulting in the following state:
\begin{equation}
\begin{split}
\ket{\Phi_6}=&\prod\limits_{w'\in\{0,1\}^{t-1}}\left[\left(I^{\otimes n-t+3\cdot BI(w'0)/2}\otimes{\rm CNOT}\otimes I^{\otimes 3\cdot2^{t-1}+3t-3\cdot BI(w'0)/2}\right)\right.\\
        &\left(I^{\otimes n-t+3\cdot BI(w'0)/2}\otimes{\rm CCNOT}\otimes I^{\otimes 3\cdot2^{t-1}+3t-3\cdot BI(w'0)/2-1}\right)\\
       &\left.\left(O''_{f_{w'1}}\otimes I^{\otimes 3\cdot2^{t-1}+3t-3\cdot BI(w'0)/2}\right)\left(O''_{f_{w'0}}\otimes I^{\otimes 3\cdot2^{t-1}+3t-3\cdot BI(w'0)/2+1}\right)\right]\\
       &\quad\prod\limits_{i=1}^2\left(I^{\otimes n-t+i}\otimes A\otimes I^{\otimes (3-i)t+2}\right)\left(I^{\otimes n-t+3\cdot 2^{t-1}}\otimes V\otimes I\right)\ket{\Phi_5}\\
=&\frac{1}{\sqrt{2^{n-t}}}\sum_{u\in\{0,1\}^{n-t}}
|u\rangle\Ket{0^{3\cdot 2^{t-1}+3t+1}}\left(\frac{\Delta(u)}{2^{t-1}}\ket{0}+\sqrt{1-\left(\frac{\Delta(u)}{2^{t-1}}\right)^2}\ket{1}\right).
\end{split} 
\end{equation}


After measurement on the last qubit of $\ket{\Phi_6}$, if the result is $1$, then $\exists$ $u\in\{0,1\}^{n-t}$ such that $|\Delta(u)|\neq 2^{t-1}$. From Corollary \ref{Cor6}, we know that $f$ is balanced. If the result is $0$, then we get the state:
\begin{equation}
\begin{split}
|\Phi_{7}\rangle=\sum_{u\in\{0,1\}^{n-t}}\frac{\Delta(u)}{\sqrt{\sum\limits_{u\in\{0,1\}^{n-t}}\Delta^2(u)}}|u\rangle\ket{0^{3\cdot 2^{t-1}+3t+2}}.
\end{split}
\end{equation}


By using Hadamard transformation on  the first $n-t$ qubits of $\ket{\Phi_7}$, we get the following state:
\begin{equation}
\begin{split}
	|\Phi_{8}\rangle=&\left(H^{\otimes n-t}\otimes I^{\otimes {3\cdot 2^{t-1}+3t+2}}\right)|\Phi_{7}\rangle\\
	=&\sum_{u,z\in\{0,1\}^{n-t}}\frac{\Delta(u)}{\sqrt{\sum\limits_{u\in\{0,1\}^{n-t}}\Delta^2(u)}}\frac{(-1)^{u\cdot z}}{\sqrt{2^{n-t}}}\ket{z}\ket{0^{3\cdot 2^{t-1}+3t+2}}.
\end{split} 
\end{equation}

The probability of measuring the first $n-t$ qubits of $\ket{\Phi_8}$ with the result of $0^{n-t}$ is
\begin{equation}
\left|\frac{1}{\sqrt{2^{n-t}\sum\limits_{u\in\{0,1\}^{n-t}}\Delta^2(u)}}\sum\limits_{u\in\{0,1\}^{n-t}}\Delta(u)\right|^2.
\end{equation}




After measurement on the first $n-t$ qubits of $\ket{\Phi_8}$, according to Theorem \ref{The7}, 
if the result is  $0^{n-t}$, then $f$ is  constant, otherwise $f$ is balanced.


\subsection{Comparison with other algorithms}
Actually,  the comparisons of Algorithm \ref{algorithm5} with the   distributed quantum algorithm for DJ problem proposed previously \cite{avron_quantum_2021},  distributed classical deterministic algorithm, and DJ algorithm  are the same as that of Algorithm \ref{algorithm2}.

In the following, we analyze and compare the three algorithms we designed.  
Compared  Algorithm \ref{algorithm3}  with  Algorithm \ref{algorithm2}, Algorithm \ref{algorithm3}  has the advantage that the number of quantum gates and qubits required by the circuit is reduced. Algorithm \ref{algorithm2}  has the advantage of scaling to multiple computing nodes. 

Compared  Algorithm \ref{algorithm2}  with Algorithm \ref{algorithm5}, Algorithm \ref{algorithm2}  has the advantage of fewer total qubits and quantum gates. Algorithm \ref{algorithm5}  has the advantage that the qubit number for realizing some single unitary operators decreases. 
For instance, the number of qubits required for the unitary operator $V$ in Algorithm \ref{algorithm5} is $3t+1$, while the number of qubits required for the unitary operator $U$ in Algorithm \ref{algorithm2} is $2^t+t+2$. 

Although the number of qubits required for the unitary operator $A$ in Algorithm \ref{algorithm5} is $3\cdot2^{t-1}+t-1$, actually,  with the help of auxiliary $2^{t-1}$ qubits, we can put  the $2^{t-1}$ control qubits $\ket{a_i}$ of  the operator $A$ together by teleportation, and replace  the operator $A$ with  the operator $A'$,  the operator $A'$ is defined as:
\begin{equation}
A'\left(\bigotimes_{i\in\{0,1\}^{t-1}}\ket{a_i}\right)|d\rangle\\
=\left(\bigotimes_{i\in\{0,1\}^{t-1}}\ket{a_i}\right)\Ket{d\oplus\sum_{i\in\{0,1\}^{t-1}}a_i},
\end{equation}
where $\forall a_i\in \{0,1\}$,  $i\in\{0,1\}^{t-1}$ and $d\in \{0,1\}^{t}$.
It is clear that the qubit number of $A'$ is $2^{t-1}+t$, which is less than the qubit number of  $A$ in Algorithm \ref{algorithm5} and qubit number of $U$ in Algorithm \ref{algorithm2}.

 The performance of our  algorithms is shown in TAB. \ref{Tab1}.

\setcounter{table}{0}

\begin{table}[H]
\begin{center}
\begin{tabular}{|l|c|c|c|}
\hline
\diagbox{\makecell[c]{Name}}{\makecell[c]{Index}} &\makecell[c]{Total number\\ of qubits} &\makecell[c]{Number of\\ quantum gates}&\makecell[c]{The number of qubits of \\ a single unitary operator} \\
\hline
 \quad Algorithm \ref{algorithm3} & $n$ & 5 & The qubit number of $Z$ is $1$ \\ 
\hline
 \quad Algorithm \ref{algorithm2} & $n+2^t+3$ & $2^{t+1}+6$ & \makecell[c]{The qubit number of $U$ is $2^t+t+2$\\ The qubit number of $R$ is $t+3$}\\  
 \hline
 \quad Algorithm \ref{algorithm5} & $n+3\cdot2^{t-1}+2t+2$ & $2^{t+2}+10$ & \makecell[c]{The qubit number of $A$ is $3\cdot2^{t-1}+t-1$\\ The qubit number of $V$ is $3t+1$\\ The qubit number of $R'$ is $t+2$} \\ 
\hline
\end{tabular}
\caption{The performance table of our algorithms.}\label{Tab1}
\end{center}
\end{table}

\section{Concluding remarks} \label{sec:conclusions}

In comparsion to quantum algorithms, dsitributed quantum algorithms usually have the advantages of less number of input qubits and circuit depth \cite{avron_quantum_2021,Elementary_gates_1995,beals_efficient_2013,Qiu2017DQC,Qiu22,Tan2022DQCSimon,Xiao2023DQAShor}, and this subject is therefore  important and practical. 
However, for designing efficient distributed quantum algorithms, the structure of problem to be solved should be clarified in distributed framework.

	In this paper, we have discovered the essential structure of DJ problem in  distributed scenario by presenting a number of equivalence characterizations between a DJ problem being constant (balanced) and the properties of its subfunctions. These structure properties have provided fundamental ideas for designing distributed exact DJ algorithms. If these structure properties are ignored and we just use DJ algorithm to solve each subfunction, then the result is not exact and the error will be higher and higher with the increasing of number of subfunctions.

By using the structure properties of DJ problem in distributed situation, we have designed three distributed exact quantum algorithms for solving DJ problem, that is,  Algorithm \ref{algorithm3}, Algorithm \ref{algorithm2} and Algorithm \ref{algorithm5}. Algorithm \ref{algorithm3} can only solve DJ problem in distributed condition with two subfunctions. So, we have  designed 
 Algorithm \ref{algorithm2}, and it can solve DJ problem in distributed situation   with multiple subfunctions.

By combining   the  ideas and methods of Algorithm \ref{algorithm3} and Algorithm \ref{algorithm2}, we have further given  Algorithm \ref{algorithm5},
which also can solve DJ problem in distributed situation with multiple computing nodes, and some of its single quantum gates require less qubits than Algorithm \ref{algorithm2}.

These distributed exact DJ algorithms we have designed have the following advantages. First,  compared with distributed classical deterministic  algorithm, our algorithms have  exponential advantage in query complexity; second,  compared with  DJ algorithm, the single query operator in our algorithms requires fewer qubits, and the depth of circuit is reduced \cite{Elementary_gates_1995}, which has better anti-noise performance.

By the way, since each oracle in Algorithm \ref{algorithm2} and Algorithm \ref{algorithm5} is controlled by the same set of control bits and is serially connected,  Algorithm \ref{algorithm2} and Algorithm \ref{algorithm5} are serial quantum query algorithms.
In fact, the oracle query of each computing node in Algorithm \ref{algorithm2} and Algorithm \ref{algorithm5} can be completed in parallel. With the help of auxiliary $\left(2^t-1\right)\left(n-t\right)$ qubits, we can change the state
 (i.e. $\dfrac{1}{\sqrt{2^{n-t}}}\sum_{u\in\{0,1\}^{n-t}}|u\rangle$) of the control register after the first Hadamard transformation to $\dfrac{1}{\sqrt{2^{n-t}}}\sum_{u\in\{0,1\}^{n-t}}\underbrace{|u\rangle|u\rangle \ldots |u\rangle}_{2^t}$. That is to say, we can change the control register from one group to the same $2^t$ groups.

More exactly, after the first Hadamard transformation, we can teleport each group of $n-t$ control qubits to every computing node and use these to control the oracles of the computing nodes.
By teleporting control qubits, we can replace  $O'_{f_w}$ in Algorithm \ref{algorithm2} and  $O''_{f_w}$ in Algorithm \ref{algorithm5} by $O_{f_w}$.

It is clear that in general  the number of qubits needed for $O_{f_w}$ is $n-t+1$, which is less than those for $O'_{f_w}$ and $O''_{f_w}$, due to the fact that $O_{f_w}$ does not necessarily have to cross the line. Therefore, using quantum teleportation to transmit control bits would not only change Algorithm \ref{algorithm2} and Algorithm \ref{algorithm5} to parallel quantum query algorithms, but also reduce the number of qubits required for each of their oracles. However, the use of quantum teleportation may increase the communication complexity of algorithms.

In sequential researches,  we would like to study distributed quantum algorithms for solving generalized DJ problem, generalized Simon problem, and other hidden group problems.




\section*{Acknowledgements}This work is supported in part by the National Natural Science Foundation of China (Nos. 61876195, 61572532), and the Natural Science Foundation of Guangdong Province of China (No. 2017B030311011).

\appendix
\section{DJ Algorithm}
\label{DJ Algorithm}

\begin{figure}[H]
	\begin{minipage}{\linewidth}	
		\begin{algorithm}[H]	
			\caption{DJ algorithm}
			\label{DJ algorithm}
			\begin{algorithmic}[1]
				\State 
				$|\psi_0\rangle=|0\rangle^{\otimes n}|1\rangle$;
							
				\State 
				$|\psi_1\rangle=H^{\otimes n+1}|\psi_0\rangle=\frac{1}{\sqrt{2^n}}\sum\limits_{x\in\{0 , 1\}^n}|x\rangle|-\rangle$;
								
				\State 			
				$|\psi_2\rangle=O_f|\psi_1\rangle=\frac{1}{\sqrt{2^n}}\sum\limits_{x\in\{0 , 1\}^n}|x\rangle\Big(\frac{|0\oplus f(x)\rangle-|1\oplus f(x)\rangle}{\sqrt{2}}\Big)=\frac{1}{\sqrt{2^n}}\sum\limits_{x\in\{0 , 1\}^n}(-1)^{f(x)}|x\rangle|-\rangle$;
							
				\State 			
				$|\psi_3\rangle=H^{\otimes n+1}|\psi_2\rangle=\sum\limits_{x,z\in\{0 , 1\}^n}\dfrac{(-1)^{x\cdot z+f(x)}}{2^n}|z\rangle|1\rangle$;
				
				\State Measure the first $n$ qubits of $\ket{\psi_3}$:  if the result is not  $0^{n}$, then output $f$ is balanced; if the result is $0^{n}$, then output $f$ is constant.
			\end{algorithmic}
		\end{algorithm}
		\end{minipage}
	\end{figure}

\begin{figure}[H]
\centering
\includegraphics[scale=0.45]{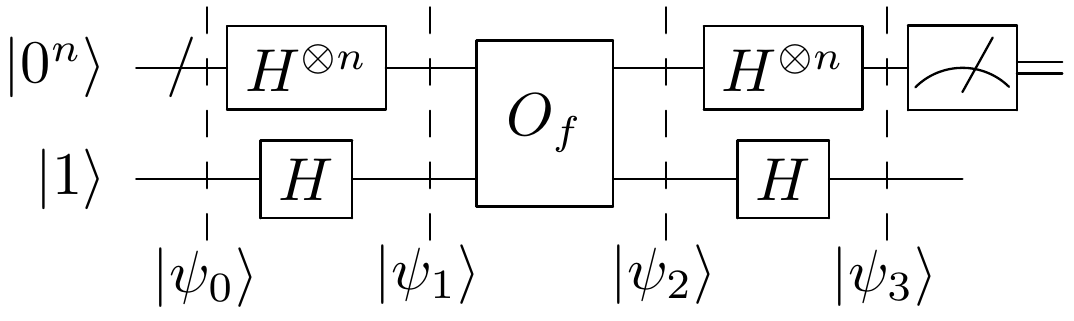}
\caption{The circuit for  DJ algorithm (Algorithm \ref{DJ algorithm} ).}
\label{fig:1}
\end{figure}

\section{Distributed DJ algorithm for multiple computing nodes with errors}
\label{Distributed DJ algorithm for multiple computing nodes with errors}

In the following, we first give the algorithm ${\rm DJ}_w$ that acts on  subfunction $f_w$.


\begin{figure}[H]
	\begin{minipage}{\linewidth}	
		\begin{algorithm}[H]	
			\caption{${\rm DJ}_w$ algorithm}
			\label{DJw algorithm}
			\begin{algorithmic}[1]
				\State 
				$|\psi'_0\rangle=|0^{n-t}\rangle|1\rangle$;
							
				\State 
				$|\psi'_1\rangle=H^{\otimes n-t+1}|\psi'_0\rangle=\frac{1}{\sqrt{2^{n-t}}}\sum\limits_{x\in\{0 , 1\}^{n-t}}|x\rangle|-\rangle$;
								
				\State 			
				$|\psi'_2\rangle=O_{f_w}|\psi'_1\rangle=\frac{1}{\sqrt{2^{n-t}}}\sum\limits_{x\in\{0 , 1\}^{n-t}}(-1)^{f_w(u)}|x\rangle|-\rangle$;
							
				\State 			
				$|\psi'_3\rangle=H^{\otimes n-t+1}|\psi'_2\rangle=\sum\limits_{x,z\in\{0 , 1\}^{n-t}}\dfrac{(-1)^{x\cdot z+f_w(u)}}{2^{n-t}}|z\rangle|1\rangle$;
				
				\State Measure the first $n-t$ qubits of $\ket{\psi'_3}$:  if the result is not  $0^{n-t}$, then output $f_w$ is not constant; if the result is $0^{n-t}$, then output $f_w$ is not balanced.
			\end{algorithmic}
		\end{algorithm}
		\end{minipage}
	\end{figure}

\begin{figure}[H]
\centering
\includegraphics[scale=0.65]{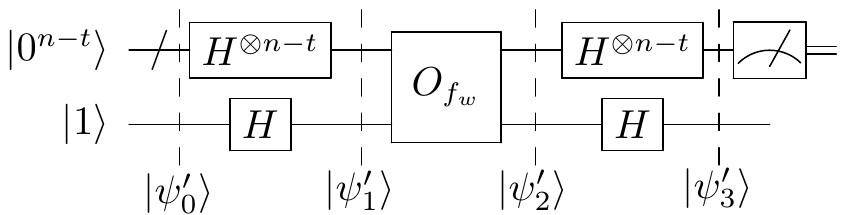}
\caption{The circuit for  ${\rm DJ}_w$ algorithm (Algorithm \ref{DJw algorithm} ).}
\label{fig:2}
\end{figure}

In the following, we design  Algorithm \ref{Distributed DJ algorithm for multiple computing nodes with errors_Algorithm flow}, which generalizes the algorithm in \cite{avron_quantum_2021} .

\begin{figure}[H]
	\begin{minipage}{\linewidth}	
		\begin{algorithm}[H]	
			\caption{Distributed DJ algorithm for multiple computing nodes with errors}
			\label{Distributed DJ algorithm for multiple computing nodes with errors_Algorithm flow}
			\begin{algorithmic}[1]
				\State 
				Decompose Boolean function $f:\{0,1\}^n\rightarrow \{0,1\}$ into $2^t$ subfunctions $f_w$ $(w\in\{0,1\}^t)$ as Equation (\ref{General method of function decomposition});
							
				\State 
				Apply the ${\rm DJ}_w$ algorithm to subfunction  $f_w$;
								
				\State 			
				Record  the measurement  result  of ${\rm DJ}_w$ algorithm as  $M_w$;
							
				\State 			
				If  there  is  $M_w$ that  is not  $0^{n-t}$,  then output  $f$ is   balanced;
				
				\State If  all  $M_w$ are $0^{n-t}$, then output  $f$  is  constant.
			\end{algorithmic}
		\end{algorithm}
		\end{minipage}
	\end{figure}


In the following, we give the error analysis for Algorithm \ref{Distributed DJ algorithm for multiple computing nodes with errors_Algorithm flow}.


Let
$l_w=|\{u\in\{0,1\}^{n-t}|f_w(u)=1\}|$,
where $w\in\{0,1\}^t$.

The probability that $M_w=0^{n-t}$ in the case $l_w=k_w$ is
\begin{equation}\label{Miprob}
\begin{split}
&\Pr(M_w=0^{n-t}| l_w=k_w)\\
=&\left|\frac{1}{2^{n-t}}[k_w\cdot(-1)+(2^{n-t}-k_w)\cdot1]\right|^2\\
=&\left(\dfrac{2^{t+1}}{N}k_w-1\right)^2.
\end{split}
\end{equation}


The  probability that Algorithm \ref{Distributed DJ algorithm for multiple computing nodes with errors_Algorithm flow}  misidentifies a balanced function as constant is 
\begin{equation}\label{Distributed DJ algorithm for multiple computing nodes with errors_Algorithm_probability}
\begin{split}
&\Pr(f\ is\ \ balanced,M_w=0^{n-t},\forall w\in\{0,1\}^t)\\
=&\sum\limits_{\substack{\sum\limits_{w\in\{0,1\}^t}k_w=\frac{N}{2}\\ 0\leq k_w\leq\frac{N}{2^t}}}\Pr(l_w=k_w)\Pr(M_w=0^{n-t},\forall w\in\{0,1\}^t |l_w=k_w)\\
=&\sum\limits_{\substack{\sum\limits_{w\in\{0,1\}^t}k_w=\frac{N}{2}\\ 0\leq k_w\leq\frac{N}{2^t}\\ }}\frac{\dbinom{N/2^t}{k_0}\dbinom{N/2^t}{k_1}\cdots\dbinom{N/2^t}{k_{2^t-1}}}{\dbinom{N}{N/2}}\prod\limits_{w\in\{0,1\}^t}\left(\dfrac{2^{t+1}}{N}k_w-1\right)^2.\\
>0.
\end{split}
\end{equation}

Although the algorithm in \cite{avron_quantum_2021} can be extended to the general case of multiple distributed computing nodes, i.e. Algorithm \ref{Distributed DJ algorithm for multiple computing nodes with errors_Algorithm flow}, it follows from equation \ref{Distributed DJ algorithm for multiple computing nodes with errors_Algorithm_probability} that the  probability that Algorithm \ref{Distributed DJ algorithm for multiple computing nodes with errors_Algorithm flow}  misidentifies a balanced function as constant is at least  $0$. 

\section{Examples of the structure of  DJ problem in  distributed scenario}
\label{Examples of the structure of  DJ problem in  distributed scenario}

\begin{example}
Given a Boolean function $f:\{0,1\}^3\rightarrow\{0,1\}$ for DJ problem, decompose $f$ into two subfunctions: $f_{0}$ and $f_{1}$, which are assumed to be as follows.
\begin{table}[H]
\begin{center}
\begin{tabular}{|l|c|c|}
\hline
\diagbox{\makecell[c]{$u$}}{\makecell[c]{$f_w(u)$}}{\makecell[c]{$w$}} &\makecell[c]{0} &\makecell[c]{1}\\
\hline
 \quad 00 &   1  &  0  \\ 
\hline
 \quad 01 &   0  &  0 \\  
 \hline
 \quad 10 &   0  &  1 \\ 
  \hline
 \quad 11 &   1  & 1 \\ 
\hline
\end{tabular}
\caption{Example of the structured table for  DJ problem in  distributed scenario with two subfunctions}\label{Tab_Example of a structured table for a DJ problem in a distributed scenario with two subfunctions}
\end{center}
\end{table}

It is clear that 
\begin{align}
B_{00}=&|\{u\in\{0,1\}^{2}|f_0(u)=f_1(u)=0\}|=1.\\
B_{11}=&|\{u\in\{0,1\}^{2}|f_0(u)= f_1(u)=1\}|=1.\\
 M=&|\{u\in\{0,1\}^{2}|f_0(u)\oplus f_1(u)=0\}|=2. 
\end{align}

Therefore 
\begin{equation}
B_{00}=B_{11}=M/2.
\end{equation}

From Theorem \ref{The4}, we can deduce that $f$ is balanced.

\end{example}

\begin{example}
Given a Boolean function $f:\{0,1\}^4\rightarrow\{0,1\}$ for DJ problem, decompose $f$ into four subfunctions: $f_{00}$, $f_{01}$, $f_{10}$ and $f_{11}$, which are assumed to be as follows.

\begin{table}[H]
\begin{center}
\begin{tabular}{|l|c|c|c|c|}
\hline
\diagbox{\makecell[c]{$u$}}{\makecell[c]{$f_w(u)$}}{\makecell[c]{$w$}} &\makecell[c]{00} &\makecell[c]{01}&\makecell[c]{10}&\makecell[c]{11} \\
\hline
 \quad 00 &   1  &  0 & 1  & 0 \\ 
\hline
 \quad 01 &   1  &  0 & 1 & 1\\  
 \hline
 \quad 10 &   0  &  1 & 0 & 0\\ 
  \hline
 \quad 11 &   1  & 1 &  0 &  0\\ 
\hline
\end{tabular}
\caption{Example of the structured table for  DJ problem in  distributed scenario with four subfunctions}\label{Tab_Example of a structured table for a DJ problem in a distributed scenario with four subfunctions}
\end{center}
\end{table}

It is clear that 
\begin{equation}
\delta(00)=0;\ \delta(01)=-2;\ \delta(10)=2;\ \delta(11)=0.
\end{equation}

Therefore 
\begin{equation}
\sum\limits_{u\in\{0,1\}^{2}} \delta(u) = 0.
\end{equation}

From Theorem \ref{The2}, we can deduce that $f$ is balanced.

\end{example}

\begin{example}
Given a Boolean function $f:\{0,1\}^4\rightarrow\{0,1\}$ for DJ problem, decompose $f$ into four subfunctions: $f_{00}$, $f_{01}$, $f_{10}$ and $f_{11}$, which are assumed to be as follows.

\begin{table}[H]
\begin{center}
\begin{tabular}{|l|c|c|c|c|}
\hline
\diagbox{\makecell[c]{$u$}}{\makecell[c]{$f_w(u)$}}{\makecell[c]{$w$}} &\makecell[c]{00} &\makecell[c]{01}&\makecell[c]{10}&\makecell[c]{11} \\
\hline
 \quad 00 &   1  &  1 & 0  & 0 \\ 
\hline
 \quad 01 &   1  &  0 & 1 & 1\\  
 \hline
 \quad 10 &   0  &  1 & 0 & 0\\ 
  \hline
 \quad 11 &   1  &  0 &  1 &  0\\ 
\hline
\end{tabular}
\caption{Example of the structured table for  DJ problem in  distributed scenario with four subfunctions}\label{Tab_Example of a structured table for a DJ problem in a distributed scenario with four subfunctions}
\end{center}
\end{table}

It is clear that 
\begin{equation}
\Delta(00)=0;\ \Delta(01)=-1;\ \Delta(10)=1;\ \Delta(11)=0.
\end{equation}

Therefore 
\begin{equation}
\sum\limits_{u\in\{0,1\}^{2}} \Delta(u) = 0.
\end{equation}

From Theorem \ref{The7}, we can deduce that $f$ is balanced.

\end{example}

\section{Distributed quantum algorithm for DJ problem with errors (four distributed computing nodes)}
\label{Distributed quantum algorithm for DJ problem with errors (four distributed computing nodes)}

\begin{figure}[H]
	\begin{minipage}{\linewidth}	
		\begin{algorithm}[H]	
			\caption{Distributed quantum algorithm for DJ problem with errors (four distributed computing nodes)}
			\label{algorithmDDJerror}
			\begin{algorithmic}[1]
				\State $|\varphi'_0\rangle = |0^{n-2}\rangle\ket{0^{n-2}}$;
				
				\State $|\varphi'_1\rangle = \left(H^{\otimes n-2}\otimes I\right)|\psi'_0\rangle$; 
				
				\State 		
				$|\varphi'_2\rangle=O_{f_{01}}O_{f_{00}}|\varphi'_1\rangle$;
				
				\State 
				$|\varphi'_3\rangle=(I\otimes Z)|\varphi'_2\rangle$;
				
				\State 		
				$|\varphi'_4\rangle=O_{f_{11}}O_{f_{10}}|\varphi'_3\rangle$;
				
				\State Measure the last qubit of $|\varphi'_4\rangle$: if the result is 1, then output $f$ is balanced; if the result is 0, then denote the quantum state after measurement  as $\ket{\varphi'_5}$;
						
				\State $|\varphi'_6\rangle=\left(H^{\otimes n-2}\otimes I\right)|\varphi'_5\rangle$;
				
				\State Measure the first $n-2$ qubits of $\ket{\varphi'_6}$:  if the result is not  $0^{n-2}$, then output $f$ is balanced; if the result is $0^{n-2}$, then output $f$ is constant.
			\end{algorithmic}
		\end{algorithm}			
		\end{minipage}
	\end{figure}

  \begin{figure}[H]
		\centering
		\includegraphics[width=6.3in]{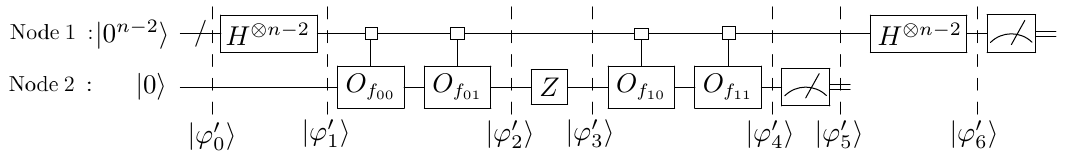}
		\caption{The circuit for the  distributed quantum algorithm for DJ problem with errors (four computing nodes) (Algorithm \ref{algorithmDDJerror} ).}
		\label{algorithmDDJerror_circuit}
	\end{figure}

In the following, we prove that  Algorithm \ref{algorithmDDJerror} is not exact and has error. The state after the first step of Algorithm \ref{algorithmDDJerror} is:
\begin{align}
  |\varphi'_1\rangle&=\frac{1}{\sqrt{2^{n-2}}}\sum_{u\in\{0,1\}^{n-2}}|u\rangle\ket{0}.
\end{align}

 Algorithm \ref{algorithmDDJerror} then queries the oracle $O_{f_{00}}$ and the oracle $O_{f_{01}}$, resulting in the following state:
\begin{equation}
\begin{split}
  |\varphi'_2\rangle&=O_{f_{01}}O_{f_{00}}|\varphi'_1\rangle\\
  &=\frac{1}{\sqrt{2^{n-2}}}\sum_{u\in\{0,1\}^{n-2}}|u\rangle|f_{00}(u)\oplus f_{01}(u)\rangle\\ 
\end{split}
\end{equation}

Then, the quantum gate $Z$ on  $|\varphi'_2\rangle$ is applied to get the following states:
\begin{equation}
\begin{split}
  |\varphi'_3\rangle=&(I\otimes Z)|\varphi'_2\rangle\\
  =&\frac{1}{\sqrt{2^{n-2}}}\sum_{u\in\{0,1\}^{n-2}}(-1)^{f_{00}(u)\oplus f_{01}(u)}|u\rangle|f_{00}(u)\oplus f_{01}(u)\rangle.
\end{split}
\end{equation}

After applying the operator $O_{f_{10}}$ and $O_{f_{11}}$ on $|\varphi'_3\rangle$, we have the following state:
\begin{equation}
\begin{split}
  |\varphi'_4\rangle=&O_{f_{11}}O_{f_{10}}|\varphi'_3\rangle\\
  =&\frac{1}{\sqrt{2^{n-2}}}\sum_{u\in\{0,1\}^{n-2}}(-1)^{f_{00}(u)\oplus f_{01}(u)}|u\rangle|f_{00}(u)\oplus f_{01}(u) \oplus f_{10}(u)\oplus f_{11}(u) \rangle\\
  =&\frac{1}{\sqrt{2^{n-2}}}\sum\limits_{\substack{u\in\{0,1\}^{n-2}\\ f_{00}(u)\oplus f_{01}(u) \oplus f_{10}(u)\oplus f_{11}(u)=0}}(-1)^{f_{00}(u)\oplus f_{01}(u)}|u\rangle\ket{0}\\
  &+\frac{1}{\sqrt{2^{n-2}}}\sum\limits_{\substack{u\in\{0,1\}^{n-2}\\ f_{00}(u)\oplus f_{01}(u) \oplus f_{10}(u)\oplus f_{11}(u)=1}}(-1)^{f_{00}(u)\oplus f_{01}(u)}|u\rangle\ket{1}.
\end{split}
\end{equation}

After measuring on the last qubit of $|\varphi'_4\rangle$, if the result is $1$, then there $\exists$ $u\in\{0,1\}^{n-2}$ such that $f(u00)\oplus f(u01)\oplus f(u11)\oplus f(u10)=1$. Similar to the proof of  Corollary \ref{Cor3}, we know that $f$ is balanced. If the result is $0$, then we get the state:
\begin{equation}
\begin{split}
|\varphi'_5\rangle=&\frac{1}{\sqrt{M'}}\sum\limits_{\substack{u\in\{0,1\}^{n-2}\\ f_{00}(u)\oplus f_{01}(u) \oplus f_{10}(u)\oplus f_{11}(u)=0}}(-1)^{f_{00}(u)\oplus f_{01}(u)}|u\rangle\ket{0}, \\
\end{split}
\end{equation}
where $M'=|\{u\in\{0,1\}^{n-2}|f_{00}(u)\oplus f_{01}(u) \oplus f_{10}(u)\oplus f_{11}(u)=0\}|$.


After Hadamard transformation on the first $n-1$ qubits of $|\varphi'_5\rangle$, we get the following state:
\begin{equation}
\begin{split}
	|\varphi'_6\rangle=&\left(H^{\otimes n-2}\otimes I\right)|\varphi'_5\rangle\\
	=&\frac{1}{\sqrt{2^{n-2}M'}}\sum\limits_{\substack{u,z\in\{0,1\}^{n-2}\\ f_{00}(u)\oplus f_{01}(u) \oplus f_{10}(u)\oplus f_{11}(u)=0}}(-1)^{f_{00}(u)\oplus f_{01}(u)+u\cdot z}\ket{z}\ket{0}.
\end{split}	
\end{equation}


The probability of measuring the first $n-2$ qubits of  $|\varphi'_6\rangle$ with the result of $0^{n-2}$ is
\begin{equation}\label{algorithmDDJerrorprobability}
\left|\frac{1}{\sqrt{2^{n-2}M'}}\sum\limits_{\substack{u\in\{0,1\}^{n-2}\\ f_{00}(u)\oplus f_{01}(u) \oplus f_{10}(u)\oplus f_{11}(u)=0}}(-1)^{f_{00}(u)\oplus f_{01}(u)}\right|^2.
\end{equation}


In the following we give an example to demonstrate that Algorithm \ref{algorithmDDJerror} cannot exactly solve  DJ problem in  distributed scenario with four computing nodes.

\begin{example}
Given a Boolean function $f:\{0,1\}^4\rightarrow\{0,1\}$ for DJ problem, decompose $f$ into four subfunctions: $f_{00}$, $f_{01}$, $f_{10}$ and $f_{11}$, which are assumed to be as follows.
\begin{table}[H]
\begin{center}
\begin{tabular}{|l|c|c|c|c|}
\hline
\diagbox{\makecell[c]{$u$}}{\makecell[c]{$f_w(u)$}}{\makecell[c]{$w$}} &\makecell[c]{00} &\makecell[c]{01}&\makecell[c]{10}&\makecell[c]{11} \\
\hline
 \quad 00 &   1  &  0 & 0  & 1 \\ 
\hline
 \quad 01 &   0  &  0 & 1 & 1\\  
 \hline
 \quad 10 &   0  &  1 & 0 & 1\\ 
  \hline
 \quad 11 &   1  & 0 &  1&  0\\ 
\hline
\end{tabular}
\caption{Example of the structured table for  DJ problem in  distributed scenario with four subfunctions}\label{Tab_Example of a structured table for a DJ problem in a distributed scenario with four subfunctions}
\end{center}
\end{table}

For  Boolean function $f$, it is clear that
\begin{equation}
\begin{split}
M'
=&|\{u\in\{0,1\}^{2}|f_{00}(u)\oplus f_{01}(u) \oplus f_{10}(u)\oplus f_{11}(u)=0\}|\\
=&4.
\end{split}
\end{equation}

For  Boolean function $f$, running Algorithm \ref{algorithmDDJerror}, according to equation (\ref{algorithmDDJerrorprobability}), yields the probability of measuring the first $n-2$ qubits of  $|\varphi'_6\rangle$ with the result of $0^{n-2}$ is
\begin{equation}\label{algorithmDDJerrorprobabilityexample}
\begin{split}
&\left|\frac{1}{\sqrt{2^{2}\cdot 4}}\sum\limits_{\substack{u\in\{0,1\}^{2}\\ f_{00}(u)\oplus f_{01}(u) \oplus f_{10}(u)\oplus f_{11}(u)=0}}(-1)^{f_{00}(u)\oplus f_{01}(u)}\right|^2\\
=&\left|\frac{1}{\sqrt{2^{2}\cdot 4}}\left[(-1)^1+(-1)^0+(-1)^1+(-1)^1\right]\right|^2\\
=&\frac{1}{4}.
\end{split}
\end{equation}

Obviously,  $f$ is  balanced. However, by equation (\ref{algorithmDDJerrorprobabilityexample}), it follows that Algorithm \ref{algorithmDDJerror}  will output $f$ is constant with probability $\frac{1}{4}$, which is wrong.

\end{example}

\end{document}